\documentclass[lettersize,journal]{IEEEtran}
\usepackage{aliascnt}
\usepackage{cite}
\usepackage{wrapfig}
\usepackage{gensymb}

\IEEEoverridecommandlockouts
\newenvironment{talign}
 {\align}
 {\endalign}
\usepackage{multirow}
\usepackage{enumitem}
\usepackage{pifont}
\usepackage{amsmath,amssymb,amsfonts,amsthm}
\usepackage{algorithmic}
\usepackage[ruled,linesnumbered]{algorithm2e}
\usepackage{graphicx}
\usepackage{textcomp}
\usepackage{caption}
\usepackage[hidelinks]{hyperref}
\usepackage{float}
\usepackage[utf8]{inputenc}
\usepackage{tabularx}
\usepackage{booktabs}
\usepackage{cases}
\usepackage{subfigure}
\usepackage{color}
\usepackage{bm}
\usepackage{extarrows}
\def\BibTeX{{\rm B\kern-.05em{\sc i\kern-.025em b}\kern-.08em
    T\kern-.1667em\lower.7ex\hbox{E}\kern-.125emX}}
    
\usepackage{graphicx}
\usepackage{mathrsfs}
\usepackage{bm}

\newtheoremstyle{slanted}
{0em plus 0em minus 0em}
  {0em plus 0em minus 0em}
  {\em}
  {}
  {\bfseries}
  {.}
  { }
  {}
\theoremstyle{slanted}

\theoremstyle{slanted}
\newtheorem{definition}{Definition}

\theoremstyle{slanted}
\newtheorem{theorem}{Theorem}

\theoremstyle{slanted}

\theoremstyle{slanted}

\theoremstyle{slanted}

\theoremstyle{slanted}
\newtheorem{lemma}{Lemma}

\usepackage{enumitem}

\title{Enhancing Data Processing Efficiency in Blockchain Enabled Metaverse over Wireless Communications}
\author{\IEEEauthorblockN{Liangxin Qian, \emph{Graduate Student Member, IEEE} and Jun Zhao, \emph{Member, IEEE}
}\thanks{The authors are with the College of Computing and Data Science (CCDS) at Nanyang Technological University, Singapore. Email: qian0080@e.ntu.edu.sg, junzhao@ntu.edu.sg.
\newline \indent
A 10-page shorter conference version \cite{qian2024data} is accepted by the ACM MobiHoc 2024. Differences between the two versions are discussed in detail in this journal submission. 
\newline \indent
This research is supported partly by Singapore Ministry of Education Academic Research Fund Tier 1 RT5/23, Tier 1 RG90/22, Nanyang Technological University (NTU)-Wallenberg AI, Autonomous Systems and Software Program (WASP) Joint Project, Imperial-Nanyang Technological University Collaboration Fund INCF-2024-008 and Fund INCF-2025-009, Seatrium New Energy Laboratory Fund Grant 03INS002302C130 (NTU-Seatrium Collaboration). This research is also supported by the National Research Foundation, Prime Minister’s Office, Singapore
under its Campus for Research Excellence and Technological Enterprise (CREATE) programme.
}}

\begin{document}
\maketitle

\begin{abstract}
In the rapidly evolving landscape of the Metaverse, enhanced by blockchain technology, the efficient processing of data has emerged as a critical challenge, especially in wireless communication systems. Addressing this challenge, our paper introduces the innovative concept of data processing efficiency (DPE), aiming to maximize processed bits per unit of resource consumption in blockchain-empowered Metaverse environments. To achieve this, we propose the \underline{D}PE-\underline{A}ware \underline{U}ser Association and \underline{R}esource Allocation (DAUR) algorithm, a tailored optimization framework for blockchain-enabled Metaverse wireless communication systems characterized by joint computing and communication resource constraints. The DAUR algorithm transforms the nonconvex problem of maximizing the sum of DPE ratios into a solvable convex optimization problem. It alternates the optimization of key variables, including user association, work offloading ratios, task-specific computing resource distribution, bandwidth allocation, user power usage ratios, and server computing resource allocation ratios. Our extensive numerical results demonstrate the DAUR algorithm's effectiveness in DPE.
\end{abstract}

\begin{IEEEkeywords}
    Blockchain, data processing efficiency, difference-of-convex programming, fractional programming, Metaverse, resource allocation, semidefinite relaxation.
\end{IEEEkeywords}

\maketitle

\section{Introduction}
The convergence of blockchain technology and the Metaverse is enabling a wide range of decentralized applications, among which non-fungible tokens (NFTs) play a prominent role \cite{wang2021non}. NFTs represent unique digital assets securely recorded on a blockchain and are increasingly used in applications \cite{nadini2021mapping}, e.g., virtual art ownership, in-game items, and digital identity. As their adoption grows, so does the demand for efficient and secure processing of NFT-related tasks, particularly in the Metaverse where real-time interaction, scalability, and transaction integrity are essential.

A central difficulty lies in the computational demands of NFT operations \cite{christodoulou2022nfts}. Verifying digital signatures, executing smart contracts, and processing blockchain blocks require considerable processing power and memory bandwidth, which are often beyond what is available on lightweight mobile or wearable user devices. Task offloading to edge or cloud infrastructure is often necessary, yet it introduces tight coupling between communication and computation processes, complicating system-wide optimization.

In parallel, the latency constraints of the Metaverse impose stringent requirements on end-to-end responsiveness \cite{van2022edge}. NFT transactions embedded in user interactions, e.g., collaborative design, asset trading, or live events, must be processed within seconds to preserve immersion. This latency is jointly influenced by wireless transmission delays, task offloading times, and blockchain verification overheads. Without coordinated resource management, these delays can accumulate, severely degrading the quality of the user experience.

Compounding this is the high energy footprint associated with blockchain-integrated systems. Maintaining distributed ledgers, executing consensus protocols, and handling cryptographic computations consume significant energy \cite{sedlmeir2020energy}, particularly problematic in wireless environments where both client and server nodes may be energy-constrained. Reducing energy consumption without compromising system performance is a potential concern in Metaverse deployments.

Another concern arises from the decision-making of user association, resource allocation, and task scheduling in existing solutions \cite{feng2020joint, dai2018joint}. Treating these aspects in isolation overlooks critical interdependencies. For instance, assigning a user to a server with high computational capacity but poor channel quality may increase delay, while allocating bandwidth without considering server availability can lead to underutilized resources. This fragmented view limits performance and scalability, especially in heterogeneous wireless Metaverse systems.

These challenges collectively motivate the need for a unified optimization framework that can jointly manage computational and communication resources. In this context, we introduce the concept of data processing efficiency (DPE), defined as the amount of data successfully processed per unit of delay and energy consumption. DPE captures the trade-off between responsiveness, energy efficiency, and computational effectiveness, making it a meaningful system-level metric for guiding optimization in blockchain-empowered environments.

This study is also driven by the imperative to improve DPE across a wide range of application settings, particularly those involving NFT-related workloads. As NFT adoption continues to expand, there is a need for frameworks that can dynamically adapt to fluctuating task complexity, user mobility, and network conditions. Supporting this demand calls for intelligent algorithms capable of performing user-server association, task offloading, communication resource allocation, and computing resource allocation in a coordinated and efficient manner.

To this end, we propose a novel optimization framework that transforms the original nonconvex DPE maximization problem into a sequence of tractable convex subproblems. Our approach integrates techniques such as fractional programming, quadratically constrained quadratic programming (QCQP), and semidefinite relaxation (SDR) with rank-1 approximation. Additionally, we design a structured alternating difference-of-convex (DC) method to address nonconvexity stemming from rank constraints, thereby ensuring convergence to high-quality solutions.

\textbf{Studied problem.}
Our research centers on a system with several users and servers, where the users delegate their NFT tasks to the servers. This offloading process is a strategic exercise in optimizing resource allocation and data processing to maximize DPE, a crucial metric in this context. DPE, representing the ratio of processed data bits to the sum of delay and energy consumption ($\frac{\textnormal{processed bits}}{\textnormal{delay} + \textnormal{energy}}$), offers a comprehensive evaluation of system performances.

We are exploring how to optimize DPE within this unique environment. Our goal is to devise a framework that not only boosts the efficiency of NFT task processing but also ensures an engaging and fluid user experience in the Metaverse, which are key features of the proposed framework. This involves addressing the intricacies of user-server connections and smart allocation of computational and communication resources in wireless networks.

\textbf{Main contributions.}
Our contributions are listed as follows:
\begin{itemize}
    \item[$\bullet$] Introduction of data processing efficiency (DPE): One contribution of this paper is the definition and exploration of the concept of DPE in blockchain-empowered Metaverse wireless communication systems. The study aims to achieve the highest possible DPE for each unit of resource consumed. This novel approach to efficiency measurement in the blockchain-Metaverse context sets a new benchmark for evaluating system performance.
    \item[$\bullet$] Development of the DAUR algorithm: We introduce the innovative \underline{D}PE \underline{A}ware \underline{U}ser association and \underline{R}esource allocation (DAUR) algorithm for blockchain-powered Metaverse wireless communications. This algorithm is a significant advancement as it simplifies the complex task of optimizing the sum of DPE ratios into a solvable convex optimization problem. A unique aspect of the DAUR algorithm is its approach to alternately optimize two sets of variables: $\{$user association, work offloading ratio, task-specific computing resource distribution$\}$ and $\{$bandwidth allocation ratio, user transmit power usage ratio, user computing resource usage ratio, server computing resource allocation ratio$\}$. By optimizing these sets together rather than separately, the DAUR algorithm achieves superior optimization results, enhancing the overall system efficiency.
    \item[$\bullet$] In the proposed DAUR algorithm, we make the discrete association variable continuous and relax the optimization of user association and work offloading ratio to a semidefinite programming (SDP) problem. Several rounding techniques and methods that solve the SDP problem with or without dropping the \hbox{rank-1} constraint are compared. Numerical results show that using different-of-convex (DC) programming while retaining the \hbox{rank-1} constraint to solve the SDP problem and using the \hbox{rank-1} approximation method to recover the discrete user association variable is an optimal and efficient way.
    \item[$\bullet$] Effectiveness of the proposed DAUR algorithm is further underscored by numerical results. These results demonstrate the algorithm's success in significantly improving DPE within the studied systems. They show the DAUR algorithm's practical utility and its potential to enhance the performance of blockchain-empowered Metaverse wireless communication systems.
\end{itemize}

\textbf{Differences with the conference version \cite{qian2024data}.} In \cite{qian2024data}, when solving Problem $\mathbb{P}_{10}$ in Section 5.2, we just drop the rank-1 constraint and then use the Hungarian algorithm augmented with zero vectors to recover the discrete variables. This method leads to sub-optimal solutions. Specifically, without the rank-1 constraint, the continuous relaxation can lead to solutions with higher rank, which may not correspond to a feasible or optimal discrete solution. Besides, when the Hungarian algorithm is used to find the optimal solution, the optimality may not translate to the original problem due to the mismatch between the continuous relaxation and the discrete problem. To improve the solution optimality when solving Problem $\mathbb{P}_{10}$, we compare two methods, i.e., dropping the rank-1 constraint and retaining the rank-1 constraint in Section \ref{sec.DAUR_algo}.E, and several rounding techniques, i.e., the Hungarian algorithm, randomized rounding, solve a secondary discrete problem, rank-1 approximation, and greedy rounding in Section \ref{sec.DAUR_algo}.D, to find the optimal and efficient way. Numerical results show that we have obtained an improvement in the solution: the DPE values of the AAUCO (introduced in Section \ref{sec.simulation_results}) and DAUR methods are improved from 81.87 and 86.48 (see Section 7 in the conference version \cite{qian2024data}) to 83.25 and 87.87 \mbox{(M bits/(s $\cdot$ J))} in Section \ref{sec.simulation_results} in this journal version, respectively. We also update the simulation results (see Section \ref{sec.simulation_results}) in this journal version submission. Besides, some important proofs are not contained in the conference version due to the length limit. In this journal version, we present those proofs in the Appendix.

The structure of this paper is outlined as follows: Section \ref{sec.related_work} presents relevant literature. Section \ref{sec.system_model} describes the system model, while the optimization problem is formulated in Section \ref{sec.optimization_prob}. The DAUR algorithm, our proposed solution for the optimization problem, is introduced in \mbox{Section \ref{sec.DAUR_algo}}, with a subsequent analysis of its complexity in Section \ref{sec.complexity_analysis}. Simulation results are presented in Section \ref{sec.simulation_results}. Conclusion and future directions are given in Section \ref{sec.conclusion}.

\section{Related work}\label{sec.related_work}
In this section, we discuss the related work on efficiency metrics and resource allocation in Blockchain and Metaverse.
\subsection{Efficiency metrics}
In wireless communications, there are a few important metrics that help us gauge how well a network performs. Spectral efficiency looks at how effectively a network uses its available bandwidth~\cite{hu2014energy}. This is especially important when the bandwidth is limited, as it tells us how much data can be transmitted within a certain frequency range. Energy efficiency, on the other hand, measures how much data can be sent for a given amount of energy \cite{wang2021lifesaving,zhou2017near,liu2021towards}. This is crucial for devices like smartphones and IoT devices, which often have limited power sources \cite{6672036}. Cost efficiency then comes into play, assessing how much data can be transmitted cost-effectively \cite{li2018cloudshare}. This balance between performance and cost is key for maintaining an efficient yet affordable network. Lastly, throughput efficiency is all about how much data can be managed in a specific area, which is vital in densely populated areas where network traffic is high \cite{ju2013throughput}.

\subsubsection{Differences between data processing efficiency and other efficiency metrics}
DPE we studied is defined as the ratio of processed data bits to the sum of delay and energy consumption. Compared to traditional metrics, DPE provides a more comprehensive evaluation by incorporating both time and energy aspects into data processing. For instance, while spectral efficiency focuses on bandwidth utilization and energy efficiency on energy per bit transmitted, DPE integrates these aspects, emphasizing the balance between delay and energy in processing data. Unlike cost efficiency measuring economic aspects, or throughput efficiency assessing data volume per area, DPE directly ties the efficiency of data processing to tangible network performance factors – delay and energy.

\subsection{Resource allocation in Blockchain systems}
Resource allocation in blockchain environments is a critical area of study, aiming to optimize various aspects of network performance under the unique constraints and opportunities presented by blockchain technology~\cite{liu2021proof}. Guo \textit{et al}. \cite{guo2019adaptive} develop a blockchain-based mobile edge computing framework that enhances throughput and Quality of Service (QoS) by optimizing spectrum allocation, block size, and number of producing blocks, using deep reinforcement learning (DRL). Deng \textit{et al}. \cite{deng2022blockchain} tackle the challenge of decentralized model aggregation in blockchain-assisted federated learning (FL), proposing a novel framework that optimizes long-term average training data size and energy consumption, employing a Lyapunov technique for dynamic resource allocation. Feng \textit{et al}.'s study \cite{feng2020joint} focuses on minimizing energy consumption and delays in a blockchain-based mobile edge computing system, though it lacks explicit detail in bandwidth and transmit power allocation. Li \textit{et al}. \cite{9983804} propose a blockchain-based IoT resource monitoring and scheduling framework that securely manages and shares idle computing resources across the network for edge intelligence tasks, ensuring reliability and fairness. Xu \textit{et al}. \cite{xu2019healthchain} introduce the Healthchain, a blockchain-based scheme for preserving the privacy of large-scale health data in IoT environments, enabling encrypted data access control and secure, tamper-proof storage of IoT data and doctor diagnoses. Finally, Liu \textit{et al}. \cite{liu2019efficient} employ game theory in blockchain-based femtocell networks to maximize the utility of users, addressing power allocation challenges. 

\subsection{Resource allocation in Metaverse systems}
Resource allocation also plays a pivotal role in forging immersive experiences within the Metaverse, a fact underscored by numerous research efforts. Zhao \textit{et al}. \cite{10368052} focus on optimizing the utility-cost ratio for Metaverse applications over wireless networks, employing a novel fractional programming technique to enhance VR video quality through optimized communication and computation resources. Meanwhile, Chu \textit{et al}. \cite{chu2023metaslicing} introduce MetaSlicing, a framework that effectively manages and allocates diverse resources by grouping applications into clusters, using a semi-Markov decision process to maximize resource utilization and Quality-of-Service. This approach dramatically improves efficiency compared to traditional methods. On the other hand, Ng \textit{et al}. \cite{ng2022unified} tackle the unified resource allocation in a virtual education setting within the Metaverse, proposing a stochastic optimal resource allocation scheme to minimize service provider costs while adapting to the users' demand uncertainty.
\section{System Model and Parameter Description}\label{sec.system_model}
The system architecture, illustrated in Fig.~\ref{fig.system_model}, consists of $N$ users and $M$ servers in a blockchain-integrated Metaverse wireless communication network, with indices $n \in \mathcal{N}:= \{1,2,\cdots,N\}$ for users and $m \in\mathcal{M}:=  \{1,2,\cdots,M\}$ for servers. Each user offloads a portion of their computational workload to a selected server to improve overall DPE. This workload is transmitted wirelessly, and servers allocate radio and computing resources to receive and process the incoming tasks.
\begin{figure}[htbp]
\centering
\includegraphics[width=0.47\textwidth]{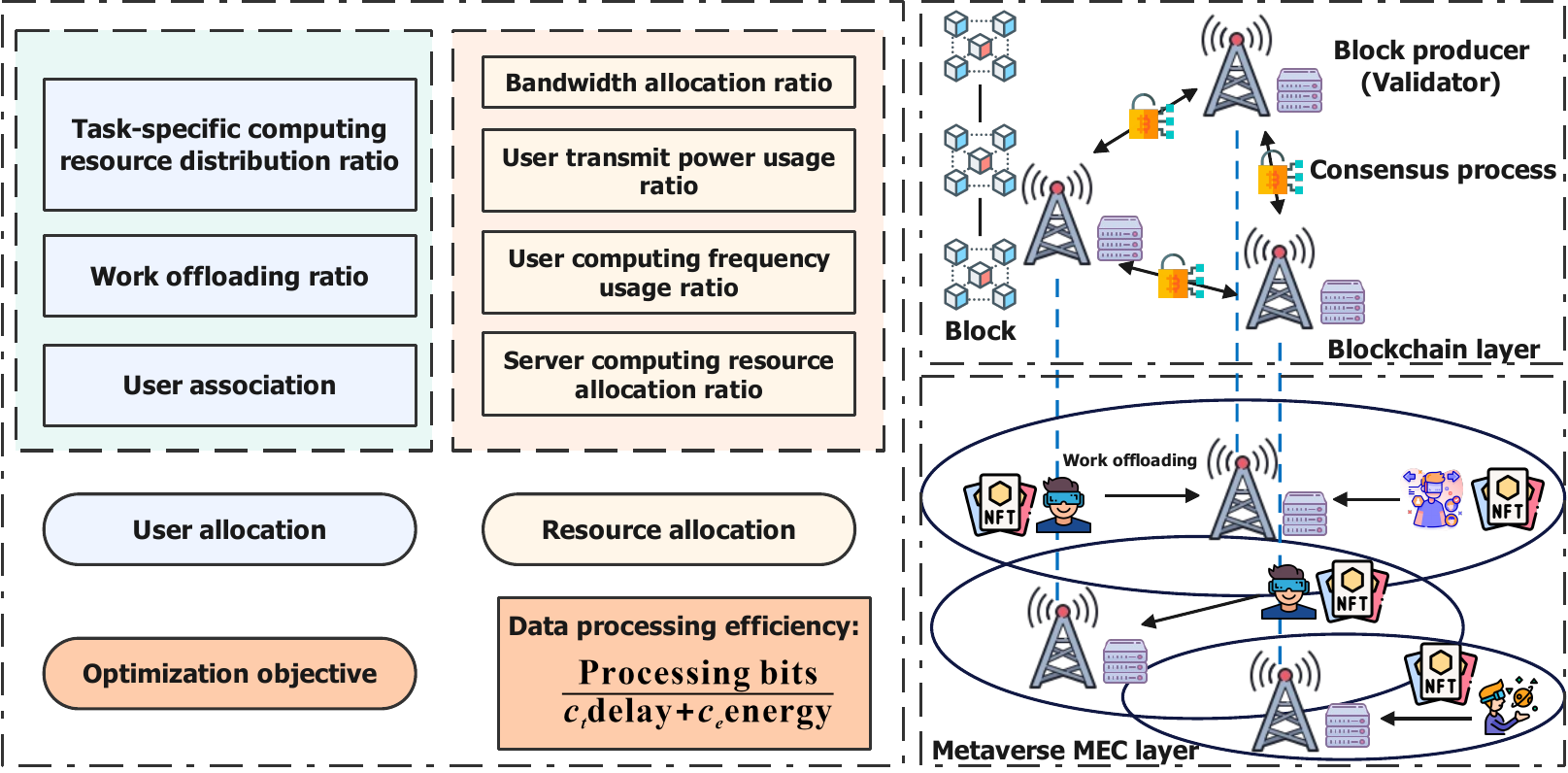}
\caption{Optimizing the DPE of a system consisting of $N$ VR users and $M$ Metaverse servers by the joint optimization of user association, offloading ratio, task-specific computing resource distribution ratio, and communication and computing resource allocation.}
\label{fig.system_model}
\end{figure}

\textbf{Explanation of Fig. \ref{fig.system_model}.} A high-level schematic of the joint optimization framework and system model is provided in \mbox{Fig. \ref{fig.system_model}}. On the left-hand side of it, we depict the optimization components involved in maximizing DPE. These include: 1) \textit{User association:} determine which server each user connects to; 2) \textit{Work offloading ratio:} the portion of each user’s task offloaded to the server; 3) \textit{Task-specific computing resource distribution ratio:} the fraction of a server's computing capacity allocated to a specific user, with internal division between user data processing and blockchain tasks; 4) \textit{Bandwidth allocation ratio:} the variables that determine the allocated bandwidth from the server to one specific user; 5) \textit{transmit power usage ratio:} the usage ratio of one user's transmitting power; 6) \textit{computing frequency usage ratio:} the usage ratio of one user's computing frequency for local data processing; 7) \textit{Server computing allocation ratio:} the variables that determine the allocated computing resources from the server to the connected user. These variables are jointly optimized to maximize DPE. We give some important notations in \mbox{Table \ref{table.notation}}.

The right-hand side represents the two-layer system architecture: 1) \textit{Metaverse MEC layer:} handle user-generated tasks, e.g., immersive VR interactions, and allow users to offload part of their workload to edge servers; 2) \textit{Blockchain layer:} ensure secure data integrity and trust. It includes the consensus process and block production operations, handled by designated validator nodes among the servers.
Tasks offloaded by users are jointly managed across both layers, requiring efficient allocation and decision making of communication and computation resources.

\begin{table}[htbp]
\caption{Important Notations.}\label{table.notation}
\centering
\setlength{\tabcolsep}{0pt}
\small
\renewcommand{\arraystretch}{1.1}
\begin{tabular}{m{7em}  m{6.5cm}}
    \toprule                                   
    \textbf{Notation}  & \textbf{Description} \\
    \midrule
    $\mathcal{N}$ & Set of all users ($n \in\{1,...,N\}$)\\ 
    $\mathcal{M}$ & Set of all servers ($m \in\{1,...,M\}$)\\
    $x_{n,m}$ & Connection between user $n$ and server $m$\\
    $d_n$ & Total work of user $n$\\
    $\varphi_n$ & Work offloading ratio of user $n$\\
    $\psi_n$ & Computation resource usage ratio of user $n$\\
    $f_n$ & Maximum computation resource of user $n$\\
    $\rho_n$ & Transmit power usage ratio of user $n$\\
    $p_n$ & Maximum transmit power of user $n$\\
    $\phi_{n,m}$ & Bandwidth allocation ratio between user $n$ and server $m$\\
    $b_m$ & Maximum allocated bandwidth for each server\\
    $\zeta_{n,m}$ & Computation resource allocation ratio from server $m$ to user $n$\\
    $f_m$ & Maximum computation resource of server $m$\\
    $\gamma_{n,m}$ & Partition ratio of computation resource server $m$ allocating to user $n$ for data processing\\
    $r_{n,m}$ & Transmission rate from user $n$ to server $m$\\
    $\omega_b$ & Data size changing ratio between the MEC task and blockchain task\\
    $\omega_t$ & Weight parameter of delay cost\\
    $\omega_e$ & Weight parameter of energy cost\\
    $S_b$ & Blockchain block size\\
    $R_m$ & Wired link transmission rate among the servers\\
    $\kappa_n$ & Effective switched capacitance of the user $n$'s chip architecture\\
    $\kappa_m$ & Effective switched capacitance of the server $m$'s chip architecture\\
    $c_n$ & DPE preference of user $n$\\
    $c_{n,m}$ & DPE preference between user $n$ and server $m$\\
    $\varpi$ & Penalty weight parameter\\
    \bottomrule
\end{tabular} 
\end{table}
\subsection{Parameters and variables in the system}
In this section, we present the specific descriptions and expressions of system parameters and optimization variables.

\subsubsection{User association}
We use $x_{n,m} \in \{0,1\}$ to denote the association between the user $n$ and the server $m$. Assume that every user is connected to one and only one server, i.e., $\sum_{m\in\mathcal{M}}x_{n,m}=1$.

\subsubsection{Partial work offloading}
Assume that $d_n$ is the total task of the user $n$. The work processed by user $n$ is $(1-\varphi_n)d_n$ and the work offloaded to server $m$ is $\varphi_n d_n$ with $\varphi_n \in [0,1]$. 

\subsubsection{Available computing resources of users}
The maximum computing resource of the user $n$ is $f_n$. We use $\psi_n \in [0,1]$ to denote the computing resource usage of the user $n$. Therefore, the computing resource used by the user $n$ to process \mbox{$(1-\varphi_n)d_n$} data is $\psi_nf_n$.

\subsubsection{Wireless communication model}
Frequency-division multiple access (FDMA) is considered in this paper so that there is no interference in wireless communication between users and servers. Assume that the total available bandwidth of the server $m$ is $b_m$. We use $\phi_{n,m}b_m$ to represent the bandwidth that the server $m$ allocates to the user $n$, where $\phi_{n,m} \in [0,1]$. Similarly, $\rho_np_n$ is used to indicate the used transmit power of the user $n$ with $\rho_n \in [0,1]$ and $p_n$ is the maximum transmit power of the user $n$. According to the Shannon formula~\cite{shannon1948mathematical}, the transmission rate between user $n$ and server $m$ is 
\begin{talign}
    r_{n,m} = \phi_{n,m}b_m\log_2(1 + \frac{g_{n,m}\rho_np_n}{\sigma^2\phi_{n,m}b_m}),
\end{talign}
where $\phi_{n,m}b_m$ is the allocated bandwidth from server $m$ to user $n$, $g_{n,m}$ is the channel attenuation between them, and $\sigma^2$ is the noise power spectral density.

\subsubsection{Task-specific computing resources distribution of servers}\label{section_partiioning_computing_resources_of_servers}
The maximum computing resource of server $m$ is $f_m$. We use $\zeta_{n,m}$ to denote the allocated computing resource part from the chosen server $m$ to the user $n$ with $\zeta_{n,m} \in [0,1]$. We assume that the server allocates part of its computing resources $\gamma_{n,m}\zeta_{n,m}f_m$ to user $n$ for data processing and another portion $(1-\gamma_{n,m})\zeta_{n,m}f_m$ for processing blockchain tasks, where $\gamma_{n,m}\in (0,1)$ and $\zeta_{n,m}f_m$ is the allocated computing resources from the chosen server $m$ to the user $n$. Since the server $m$ needs to process data and also complete the blockchain task, it must reserve a portion of computing resources for these two parts (i.e., $\gamma_{n,m} \neq 0$ or $1$).

\subsection{Cost analysis}
In this section, we will analyze the user and server costs of the uplink. The connection $x_{n,m}$ between the user $n$ and the server $m$ can be decided based on the maximization of data processing efficiency, which will be detailed in Definition \ref{def_dpe}.
\subsubsection{User cost analysis}
The local data bits for processing are $(1-\varphi_n)d_n$ and the local processing delay is  
\begin{talign}
    T^{(up)}_n = \frac{(1-\varphi_n)d_n\eta_n}{\psi_nf_n},\label{eq_T_up}
\end{talign}
where $\eta_n$ (cycles/bit) is the number of CPU cycles per bit of the user $n$ and the superscript ``$up$'' means user processing. Then, we study the energy consumption of the user side. If the available computing resources of user $n$ is $\psi_nf_n$, the local processing energy consumption of the user $n$ is 
\begin{talign}
    E^{(up)}_n = \kappa_n(1-\varphi_n)d_n\eta_n(\psi_nf_n)^2,
\end{talign}
where $\kappa_n$ is the effective switched capacitance of the user $n$. After processing the local data bits, the user $n$ needs to transmit the rest $\varphi_n d_n$ data bits to the connected server $m$ over the wireless channel. The transmission delay of the user $n$ is 
\begin{talign}
    T^{(ut)}_{n,m} = \frac{x_{n,m}\varphi_nd_n}{r_{n,m}},
\end{talign} 
where the superscript ``$ut$'' represents user transmission. Assume that the available transmit power of the user $n$ is $\rho_np_n$, the energy consumption during the transmission of the user $n$ is 
\begin{talign}
    E^{(ut)}_{n,m} = \rho_np_nT^{(ut)}_{n,m} = \frac{x_{n,m}\rho_np_n\varphi_nd_n}{r_{n,m}}.
\end{talign}

\subsubsection{Server cost analysis}
After $\varphi_n d_n$ data bits are transmitted to the server $m$ from the user $n$, the server $m$ would first process those data bits. According to the discussion in section \ref{section_partiioning_computing_resources_of_servers}, the processing delay of the server $m$ for the offloaded work from the user $n$ is 
\begin{talign}
    T^{(sp)}_{n,m} = \frac{x_{n,m}\varphi_nd_n\eta_m}{\gamma_{n,m}\zeta_{n,m}f_m},
\end{talign}
and its related data processing energy consumption is 
\begin{talign}
    E^{(sp)}_{n,m} = \kappa_mx_{n,m}\varphi_nd_n\eta_m(\gamma_{n,m}\zeta_{n,m}f_m)^2,
\end{talign}
where the superscript ``$sp$'' stands for server processing and $\eta_m$ is the number of CPU cycles per bit of the server $m$. When $\varphi_n d_n$ data bits are processed by the server $m$, it will generate a blockchain block for them. Assume that in this phase, the server $m$ would process $\varphi_nd_n\omega_b$ bits, where $\omega_b$ is the data size changing ratio of servers mapping the offloaded data to the format that will be processed by the blockchain. Since $\gamma_{n,m}$ fraction of computing resources $\zeta_{n,m}f_m$ that allocated to the user $n$ is used in the data processing phase, there is still $1-\gamma_{n,m}$ fraction of that for block generation phase. Therefore, the delay of the server $m$ generating blockchain block for the user $n$ is 
\begin{talign}
    T^{(sg)}_{n,m} = \frac{x_{n,m}\varphi_nd_n\omega_b\eta_m}{(1-\gamma_{n,m})\zeta_{n,m}f_m},
\end{talign}
and its corresponding energy consumption is 
\begin{talign}
    E^{(sg)}_{n,m} = \kappa_mx_{n,m}\varphi_nd_n\eta_m \omega_b [(1-\gamma_{n,m})\zeta_{n,m}f_m]^2,
\end{talign}
where ``$sg$'' indicates server generation. Next, we discuss the block propagation and validation cost. Assume that there is only one hop among servers. We consider only the block propagation delay and the validation delay, as in \cite{feng2020joint}. According to the analysis in \cite{feng2020joint}, we can obtain the total block propagation delay in the data transactions during the consensus as
\begin{talign}
    T^{(bp)}_{n,m} = \frac{S_b}{R_m},
\end{talign}
where $R_m:=\text{min}_{m^\prime\in\mathcal{M}\setminus\{m\}}R_{m,m^\prime}$ and $R_{m,m^\prime}$ is the wired link transmission rate between the servers $m$ and $m^\prime$ and the superscript ``\textit{bp}'' means block propagation. Besides, based on \cite{feng2020joint}, the validation time during the block propagation is 
\begin{align}
    T^{(sv)}_{n,m} = \text{max}_{m^\prime\in\mathcal{M}\setminus\{m\}}\frac{\eta_v}{(1-\gamma_{n,m^\prime})\zeta_{n,m^\prime}f_m^\prime},
\end{align}
where $\eta_v$ is the number of CPU cycles required by server $m^\prime$ to verify the block and the superscript ``$sv$'' is server validation.

\section{Definition of Data Processing Efficiency and Studied Optimization Problem}\label{sec.optimization_prob}
Here we first introduce the concept of data processing efficiency before formulating the optimization problem.
\begin{definition}[Data Processing Efficiency]\label{def_dpe}
    Data processing efficiency (DPE) refers to the costs including delay and energy required to process data bits, i.e., DPE $:= \frac{\text{processing data bits}}{\text{delay + energy}}$. This formulation reflects a trade-off in mobile and edge computing systems: low delay-energy consumption and high processing utility are both desirable, and their combination more accurately captures the system’s resource burden and performance. By placing the sum of delay and energy in the denominator, we effectively measure how efficiently a system processes data with respect to its total resource expenditure.
    
    In the two-level system (i.e., user and server), the cost of processing data at the user level includes the local data processing costs. The cost of processing data at the server level consists of both the data processing consumption and the wireless transmission consumption from the user level. For instance, if $\varphi_n d_n$ parameters are processed at the server $m$, the cost for processing these bits encompasses the delay and energy consumed in processing, as well as those involved in data transmission from the \mbox{user $n$}.
\end{definition}
Based on the concept of DPE, we obtain the cost required to process $(1-\varphi_n)d_n$ at the user $n$:
\begin{talign}
    cost^{(u)}_n = \omega_t T^{(up)}_n + \omega_e E^{(up)}_n, \label{eq.cost_u}
\end{talign}
where $\omega_t$ and $\omega_e$ denote the weight parameters of delay and energy consumption. Similarly, we get the cost needed to process $\varphi_n d_n$ at the server $m$:
\begin{talign}
    cost^{(s)}_{n,m} &= \omega_t (T^{(ut)}_{n,m} + T^{(sp)}_{n,m} + T^{(sg)}_{n,m} + T^{(bp)}_{n,m} + T^{(sv)}_{n,m}) \nonumber \\ 
    &+ \omega_e (E^{(ut)}_{n,m} + E^{(sp)}_{n,m} + E^{(sg)}_{n,m}). \label{eq.cost_s}
\end{talign}
Then, our optimization objective is to maximize the total DPE of all servers and users. Let $\bm{x}:=[x_{n,m}]|_{n\in\mathcal{N},m\in\mathcal{M}}$, $\bm{\varphi}:=[\varphi_n]|_{n\in\mathcal{N}}$, $\bm{\gamma}:=[\gamma_{n,m}]|_{n\in\mathcal{N},m\in\mathcal{M}}$, $\bm{\phi}:=[\phi_{n,m}]|_{n\in\mathcal{N},m\in\mathcal{M}}$, $\bm{\rho}:=[\rho_n]|_{n\in\mathcal{N}}$, $\bm{\zeta}:=[\zeta_{n,m}]|_{n\in\mathcal{N},m\in\mathcal{M}}$, and $\bm{\psi}:=[\psi_n]|_{n\in\mathcal{N}}$. We give the optimization problem $\mathbb{P}_{1}$ as follows:
\begin{subequations}\label{prob1}
\begin{align}
\mathbb{P}_{1}:&\!\!\!\!\max\limits_{\bm{x},\bm{\varphi},\bm{\gamma},\bm{\phi},\bm{\rho},\bm{\zeta},\bm{\psi}}  \!\!\sum\limits_{n \in \mathcal{N}} \!\!\!\frac{c_n(1-\varphi_n)d_n}{cost^{(u)}_n}\!\! +\!\!\! \sum\limits_{n \in \mathcal{N}}\sum\limits_{m \in \mathcal{M}}\!\!\frac{c_{n,m}x_{n,m}\varphi_nd_n}{cost^{(s)}_{n,m}}
\tag{\ref{prob1}}\\
\text{s.t.} \quad 
& x_{n,m} \in \{0,1\}, ~\forall n \in \mathcal{N}, ~\forall m \in \mathcal{M},\label{x_constr1} \\
& \sum\limits_{m \in \mathcal{M}} x_{n,m} = 1, ~\forall n \in \mathcal{N},\label{x_constr2} \\
& \varphi_n \in [0,1], ~\forall n \in \mathcal{N},\label{varphi_constr}\\
&\gamma_{n,m} \in (0,1), ~\forall n \in \mathcal{N}, ~\forall m \in \mathcal{M},\label{gamma_constr}\\
&\phi_{n,m} \in [0,1], ~\forall n \in \mathcal{N}, ~\forall m \in \mathcal{M},\label{phi_constr1}\\
&\sum\limits_{n\in \mathcal{N}} x_{n,m} \phi_{n,m} \leq 1, ~\forall m \in \mathcal{M},\label{phi_constr2}\\
&\zeta_{n,m} \in [0,1], ~\forall n \in \mathcal{N}, ~\forall m \in \mathcal{M},\label{zeta_constr1}\\
&\sum\limits_{n\in \mathcal{N}} x_{n,m} \zeta_{n,m} \leq 1, ~\forall m \in \mathcal{M},\label{zeta_constr2}\\
&\rho_n \in [0,1], ~\forall n \in \mathcal{N},\label{rho_constr}\\
&\psi_n \in [0,1], ~\forall n \in \mathcal{N},\label{psi_constr}
\end{align}
\end{subequations}
where $c_n$ is the DPE preference of the user $n$, $c_{n,m}$ is that between the server $m$ and the user $n$, and ``s.t.'' is short for ``subject to''. Constraints (\text{\ref{x_constr1}}) and (\text{\ref{x_constr2}}) mean that user association variable $x_{n,m}$ is 0, 1 discrete variable, and every user is connected to one and only one server. Constraints (\text{\ref{varphi_constr}}), (\text{\ref{gamma_constr}}), and (\text{\ref{phi_constr1}}) are the feasible value ranges of partial work offloading ratio $\varphi_n$, task-specific computing resources distribution ratio of servers $\gamma_{n,m}$, and available bandwidth ratio $\phi_{n,m}$. Constraint (\text{\ref{phi_constr2}}) means that the total allocated bandwidth can't be greater than the maximum available bandwidth. Constraint (\text{\ref{zeta_constr1}}) is the feasible value range of the allocated computing resource ratio $\zeta_{n,m}$ from the server. Constraint (\text{\ref{zeta_constr2}}) means that the total allocated computing resource can't be greater than the maximum available computing resource from the server. Constraints (\text{\ref{rho_constr}}) and (\text{\ref{psi_constr}}) are the feasible value ranges of transmit power usage ratio $\rho_n$ of the user and computing resource usage ratio $\psi_n$ of the user.

\section{Proposed DAUR Algorithm to Solve the Optimization Problem}\label{sec.DAUR_algo}
Problem $\mathbb{P}_{1}$ is a sum of ratios problem coupled by many complicated non-convex terms, which is generally NP-complete and difficult to solve directly. We will make Problem $\mathbb{P}_{1}$ a solvable convex problem within a series of transformations by our proposed \underline{D}PE-\underline{A}ware \underline{U}ser association and \underline{R}esource allocation (DAUR) algorithm. 
\begin{theorem}\label{theorem_solvep1}
    Problem $\mathbb{P}_{1}$ can be relaxed into a solvable problem if we alternatively optimize $[\bm{x},\bm{\varphi},\bm{\gamma}]$ and $[\bm{\phi},\bm{\rho},\bm{\zeta},\bm{\psi}]$.
\end{theorem}
\begin{proof}
    \textbf{Theorem \ref{theorem_solvep1}} is proven by the following \textbf{Lemma \ref{lemma_p1top2}}, \textbf{Lemma \ref{lemma_p2top3}}, \textbf{Theorem \ref{theorem_p4toconcave}} in Section \ref{sec.AO_step1}, and \textbf{Theorem \ref{theorem_p7toconvex}} in Section \ref{sec.AO_step2}.
\end{proof}
\begin{lemma}\label{lemma_p1top2}
Define new auxiliary variables $\vartheta_n^{(u)}$, $\vartheta_{n,m}^{(s)}$, $T_n^{(u)}$, and $T_{n,m}^{(s)}$. Let $\bm{\vartheta^{(u)}}:=[\vartheta_n^{(u)}]|_{n\in\mathcal{N}}$, $\bm{\vartheta^{(s)}}:=[\vartheta_{n,m}^{(s)}]|_{n\in\mathcal{N},m\in\mathcal{M}}$, $\bm{T^{(u)}}:=[T_n^{(u)}]|_{n\in\mathcal{N}}$, $\bm{T^{(s)}}:=[T_{n,m}^{(s)}]|_{n\in\mathcal{N},m\in\mathcal{M}}$, $\bm{T}:=\{\bm{T^{(u)}},\bm{T^{(s)}}\}$, and $\bm{\vartheta}:=\{\bm{\vartheta^{(u)}},\bm{\vartheta^{(s)}}\}$. The sum of ratios Problem $\mathbb{P}_{1}$ can be relaxed into a summation Problem $\mathbb{P}_{2}$:
\begin{subequations}\label{prob2}
\begin{align}
&\mathbb{P}_{2}:\max\limits_{\bm{x},\bm{\varphi},\bm{\gamma},\bm{\phi},\bm{\rho},\bm{\zeta},\bm{\psi}, \bm{\vartheta}, \bm{T}}  \sum\limits_{n \in \mathcal{N}} \vartheta_n^{(u)} + \sum\limits_{n \in \mathcal{N}}\sum\limits_{m \in \mathcal{M}}\vartheta_{n,m}^{(s)} \tag{\ref{prob2}}\\
&\text{s.t.} \quad (\text{\ref{x_constr1}})\text{-}(\text{\ref{psi_constr}}), \nonumber \\
& \omega_t T^{(u)}_n + \omega_e \kappa_n(1-\varphi_n)d_n\eta_n(\psi_nf_n)^2 - \frac{c_n(1-\varphi_n)d_n}{\vartheta_n^{(u)}}\leq 0, \label{vartheta_u_constr}\\
& \omega_t T^{(s)}_{n,m} - \frac{c_{n,m}x_{n,m}\varphi_nd_n}{\vartheta_{n,m}^{(s)}} + \omega_e \{\frac{x_{n,m}\rho_np_n\varphi_nd_n}{r_{n,m}} \nonumber \\ 
    &+ \kappa_mx_{n,m}\varphi_nd_n\eta_m(\gamma_{n,m}\zeta_{n,m}f_m)^2 \nonumber \\ 
    &+ \kappa_mx_{n,m}\varphi_nd_n\eta_m \omega_b [(1-\gamma_{n,m})\zeta_{n,m}f_m]^2\}  \leq 0, \label{vartheta_s_constr}\\
&T^{(up)}_n \leq T^{(u)}_n, \label{Tu_constr}\\
&T^{(ut)}_{n,m} + T^{(sp)}_{n,m} + T^{(sg)}_{n,m} + T^{(bp)}_{n,m} + T^{(sv)}_{n,m} \leq T^{(s)}_{n,m}. \label{Ts_constr}
\end{align}
\end{subequations}
\end{lemma}
\begin{proof}
    Refer to Appendix \ref{append_lemma_p1top2}.
\end{proof}
According to \textbf{Lemma \ref{lemma_p1top2}}, we can relax the sum of ratios Problem $\mathbb{P}_{1}$ to a summation Problem $\mathbb{P}_{2}$ by adding the extra auxiliary variables $\vartheta_n^{(u)}$, $\vartheta_{n,m}^{(s)}$, $T_n^{(u)}$, and $T_{n,m}^{(s)}$. Thanks to $\vartheta_n^{(u)}$ and $\vartheta_{n,m}^{(s)}$, we convert the sum of ratios of the objective function in \mbox{Problem $\mathbb{P}_{1}$} to the sum of two variables. Besides, we can transfer the troublesome terms about the delay of the objective function in \mbox{Problem $\mathbb{P}_{1}$} into the constraints (\ref{Tu_constr}) and (\ref{Ts_constr}) by introducing the variables $T_n^{(u)}$ and $T_{n,m}^{(s)}$. However, the constraints (\ref{vartheta_u_constr}) and (\ref{vartheta_s_constr}) are not convex and \mbox{Problem $\mathbb{P}_{2}$} is still hard to solve.

\begin{lemma}\label{lemma_p2top3}
Define new auxiliary variables $\alpha^{(u)}_n$ and $\alpha^{(s)}_{n,m}$. Let $\bm{\alpha^{(u)}}:=[\alpha^{(u)}_n]|_{n\in\mathcal{N}}$, $\bm{\alpha^{(s)}}:=[\alpha^{(s)}_{n,m}]|_{n\in\mathcal{N},m\in\mathcal{M}}$, and $\bm{\alpha}:=\{\bm{\alpha^{(u)}}$, $\bm{\alpha^{(s)}}\}$. After that, Problem $\mathbb{P}_{2}$ can be transformed into Problem $\mathbb{P}_{3}$:
\begin{subequations}\label{prob3}
\begin{align}
&\mathbb{P}_{3}:\max\limits_{\bm{x},\bm{\varphi},\bm{\gamma},\bm{\phi},\bm{\rho},\bm{\zeta},\bm{\psi}, \bm{\vartheta}, \bm{\alpha}, \bm{T}}  \nonumber \\
&\sum\limits_{n \in \mathcal{N}} \alpha^{(u)}_n [c_n(1 - \varphi_n)d_n - \vartheta_n^{(u)}cost^{(u)}_n] \nonumber \\
&+ \sum\limits_{n \in \mathcal{N}}\sum\limits_{m \in \mathcal{M}}\alpha^{(s)}_{n,m} (c_{n,m}x_{n,m}\varphi_n d_n - \vartheta_{n,m}^{(s)}cost^{(s)}_{n,m})\tag{\ref{prob3}}\\
&\text{s.t.} \quad (\text{\ref{x_constr1}})\text{-}(\text{\ref{psi_constr}}), (\text{\ref{Tu_constr}})\text{-}(\text{\ref{Ts_constr}}). \nonumber
\end{align}
\end{subequations}
At Karush-–Kuhn–-Tucker (KKT) points of Problem $\mathbb{P}_{3}$, we can obtain that 
\begin{talign}
    &\alpha^{(u)}_n = \frac{1}{cost^{(u)}_n},\label{eq_alpha_u}\\
    &\alpha^{(s)}_{n,m} = \frac{1}{cost^{(s)}_{n,m}},\label{eq_alpha_s}\\
    &\vartheta_n^{(u)} = \frac{c_n(1 - \varphi_n)d_n}{cost^{(u)}_n},\label{eq_vartheta_u}\\
 \text{and }   &\vartheta_{n,m}^{(s)} = \frac{c_{n,m}x_{n,m}\varphi_n d_n}{cost^{(s)}_{n,m}}.\label{eq_vartheta_s}
\end{talign}
\end{lemma}
\begin{proof}
    Refer to Appendix \ref{append_lemma_p2top3}.
\end{proof}

Based on \textbf{Lemma \ref{lemma_p2top3}}, we can split the ratio form of the objective function in Problem $\mathbb{P}_{1}$ and transform the non-convex constraints (\ref{vartheta_u_constr}) and (\ref{vartheta_s_constr}) into the objective function in Problem $\mathbb{P}_{3}$ by introducing new auxiliary variables $\alpha^{(u)}_n$ and $\alpha^{(s)}_{n,m}$. Besides, based on the analysis of the KKT conditions of Problem $\mathbb{P}_{3}$, we can obtain the relationships between auxiliary variables $[\alpha^{(u)}_n$, $\alpha^{(s)}_{n,m}$, $\vartheta_n^{(u)}$, $\vartheta_{n,m}^{(s)}]$ and original variables $[x_{n,m}$, $\varphi_n$, $\gamma_{n,m}$, $\phi_{n,m}$, $\rho_n$, $\zeta_{n,m}$, $\psi_n$, $T^{(u)}_n$, $T^{(s)}_{n,m}]$ as Equations (\ref{eq_alpha_u}), (\ref{eq_alpha_s}), (\ref{eq_vartheta_u}), (\ref{eq_vartheta_s}). Next, we consider alternative optimization to solve this complex problem. At the \mbox{$i$-th} iteration, we first fix $\bm{\alpha}^{(i-1)}$ and $\bm{\vartheta}^{(i-1)}$, and then optimize $\bm{x}^{(i)},\bm{\varphi}^{(i)},\bm{\gamma}^{(i)},\bm{\phi}^{(i)},\bm{\rho}^{(i)},\bm{\zeta}^{(i)},\bm{\psi}^{(i)}, \bm{T}^{(i)}$. We then update $\bm{\alpha}^{(i)}$ and $\bm{\vartheta}^{(i)}$ according to their results. Repeat the above optimization steps until the relative difference between the objective function values of Problem $\mathbb{P}_{3}$ in the $i$-th and \mbox{$(i-1)$-th} iterations is less than an acceptable threshold, and we get a stationary point for Problem $\mathbb{P}_{3}$. Next, we analyze how to optimize $\bm{x},\bm{\varphi},\bm{\gamma},\bm{\phi},\bm{\rho},\bm{\zeta},\bm{\psi},\bm{T}$ with the given $\bm{\vartheta}, \bm{\alpha}$.

\subsection{Solve \texorpdfstring{$\bm{\phi},\bm{\rho},\bm{\zeta},\bm{\psi}$}{}, and \texorpdfstring{$\bm{T}$}{} with fixed \texorpdfstring{$\bm{x}$}{}, \texorpdfstring{$\bm{\gamma}$}{}, and \texorpdfstring{$\bm{\varphi}$}{}}\label{sec.AO_step1}
If we first fix $\bm{x}$, $\bm{\gamma}$ and $\bm{\varphi}$, Problem $\mathbb{P}_{3}$ would be the following new Problem $\mathbb{P}_{4}$:
\begin{subequations}\label{prob4}
\begin{align}
\mathbb{P}_{4}:&\max\limits_{\bm{\phi},\bm{\rho},\bm{\zeta},\bm{\psi}, \bm{T}}  \sum\limits_{n \in \mathcal{N}} \alpha^{(u)}_n [c_n(1 - \varphi_n)d_n - \vartheta_n^{(u)}cost^{(u)}_n] \nonumber \\
&+ \!\!\sum\limits_{n \in \mathcal{N}}\sum\limits_{m \in \mathcal{M}}\alpha^{(s)}_{n,m} (c_{n,m}x_{n,m}\varphi_n d_n - \vartheta_{n,m}^{(s)}cost^{(s)}_{n,m}) \tag{\ref{prob4}}\\
&\text{s.t.} \quad (\text{\ref{phi_constr1}})\text{-}(\text{\ref{psi_constr}}), (\text{\ref{Tu_constr}})\text{-}(\text{\ref{Ts_constr}}). \nonumber
\end{align}
\end{subequations}
Note that $cost^{(s)}_{n,m}$ is still not convex. 
\begin{theorem}\label{theorem_p4toconcave}
    Problem $\mathbb{P}_{4}$ can be transformed into a solvable concave optimization problem by a fractional programming (FP) technique.
\end{theorem}
\begin{proof}
   \textbf{Theorem \ref{theorem_p4toconcave}} is proven by the following \mbox{\textbf{Lemma \ref{lemma_p4top5}}}.
\end{proof}

\begin{lemma}\label{lemma_p4top5}
Define a new auxiliary variable $\upsilon^{(s)}_{n,m}$, where $\upsilon^{(s)}_{n,m} = \frac{1}{2x_{n,m}\rho_np_n\varphi_nd_nr_{n,m}}$. Rewrite $cost^{(s)}_{n,m}$ as a new variable $\widetilde{cost}^{(s)}_{n,m}$ with $\upsilon^{(s)}_{n,m}$. Let $\bm{\upsilon^{(s)}} := [\upsilon^{(s)}_{n,m}|_{\forall n \in \mathcal{N},\forall m \in \mathcal{M}}]$. The Problem $\mathbb{P}_{4}$ can be transformed into the following Problem $\mathbb{P}_{5}$:
\begin{subequations}\label{prob5}
\begin{align}
\mathbb{P}_{5}:&\max\limits_{\bm{\phi},\bm{\rho},\bm{\zeta},\bm{\psi}, \bm{\upsilon^{(s)}},\bm{T}}  \sum\limits_{n \in \mathcal{N}} \alpha^{(u)}_n [c_n(1 - \varphi_n)d_n - \vartheta_n^{(u)}cost^{(u)}_n] \nonumber \\
&+ \!\!\sum\limits_{n \in \mathcal{N}}\sum\limits_{m \in \mathcal{M}}\alpha^{(s)}_{n,m} (c_{n,m}x_{n,m}\varphi_n d_n - \vartheta_{n,m}^{(s)}\widetilde{cost}^{(s)}_{n,m}) \tag{\ref{prob5}}\\
&\text{s.t.} \quad (\text{\ref{phi_constr1}})\text{-}(\text{\ref{psi_constr}}),(\text{\ref{Tu_constr}})\text{-}(\text{\ref{Ts_constr}}). \nonumber
\end{align}
\end{subequations}
If we alternatively optimize $\bm{\upsilon^{(s)}}$ and $\bm{\phi},\bm{\rho},\bm{\zeta},\bm{\psi},\bm{T}$, Problem $\mathbb{P}_{5}$ would be a concave problem. Besides, with the local optimum $\bm{\upsilon}^{\bm{(s)}(\star)}$, we can find $\bm{\phi}^{(\star)},\bm{\rho}^{(\star)},\bm{\zeta}^{(\star)},\bm{\psi}^{(\star)}, \bm{T}^{(\star)}$, which is a stationary point of Problem $\mathbb{P}_{5}$.
\end{lemma}
\begin{proof}
    Refer to Appendix \ref{append_lemma_p4top5}.
\end{proof}
Thanks to \textbf{Lemma \ref{lemma_p4top5}}, we can transform Problem $\mathbb{P}_{4}$ into a solvable concave problem Problem $\mathbb{P}_{5}$ with the help of $\upsilon^{(s)}_{n,m}$. This transformation enables the use of an FP-based method, where the original nonconvex ratio structure is iteratively approximated through the following procedures.

During the $i$-th iteration, we initially hold $\bm{\upsilon}^{\bm{(s)}(i-1)}$ constant and focus on optimizing $\bm{\phi}^{(i)}, \bm{\rho}^{(i)}, \bm{\zeta}^{(i)}, \bm{\psi}^{(i)}, \bm{T}^{(i)}$. Once these values are determined, we update $\bm{\upsilon}^{\bm{(s)}(i)}$ based on the obtained results. This optimization cycle is repeated until the difference in the objective function value of Problem $\mathbb{P}_{5}$ between the $i$-th and $(i-1)$-th iterations falls in a predefined threshold. Reaching this point signifies a solution for Problem $\mathbb{P}_{5}$, and consequently, for Problem $\mathbb{P}_{4}$. Next, we analyze how to optimize $\bm{x}$, $\bm{\gamma}$, $\bm{\varphi}$, and $\bm{T}$ with fixed $\bm{\phi},\bm{\rho},\bm{\zeta}$, and $\bm{\psi}$.

\subsection{Solve \texorpdfstring{$\bm{x}$}{}, \texorpdfstring{$\bm{\gamma}$}{}, \texorpdfstring{$\bm{\varphi}$}{}, and  \texorpdfstring{$\bm{T}$}{} with fixed \texorpdfstring{$\bm{\phi},\bm{\rho},\bm{\zeta}$}{}, and \texorpdfstring{$\bm{\psi}$}{}}\label{sec.AO_step2}
If $\bm{\phi},\bm{\rho},\bm{\zeta}$, and $\bm{\psi}$ are given, Problem $\mathbb{P}_{3}$ would be the following Problem:
\begin{subequations}\label{prob6}
\begin{align}
\mathbb{P}_{6}:&\max\limits_{\bm{x},\bm{\varphi},\bm{\gamma},\bm{T}}  \sum\limits_{n \in \mathcal{N}} \alpha^{(u)}_n [c_n(1 - \varphi_n)d_n - \vartheta_n^{(u)}cost^{(u)}_n] \nonumber \\
&+ \!\!\sum\limits_{n \in \mathcal{N}}\sum\limits_{m \in \mathcal{M}}\alpha^{(s)}_{n,m} (c_{n,m}x_{n,m}\varphi_n d_n - \vartheta_{n,m}^{(s)}cost^{(s)}_{n,m}) \tag{\ref{prob6}}\\
&\text{s.t.} \quad (\text{\ref{x_constr1}})\text{-}(\text{\ref{gamma_constr}}),~(\text{\ref{phi_constr2}}),~(\text{\ref{zeta_constr2}}),~(\text{\ref{Tu_constr}})\text{-}(\text{\ref{Ts_constr}}). \nonumber
\end{align}
\end{subequations}
Since $x_{n,m}$ is a binary discrete variable and other variables are continuous, Problem $\mathbb{P}_{6}$ is a mixed-integer nonlinear programming problem. We transform the constraint (\text{\ref{x_constr1}}) to $x_{n,m}(x_{n,m}-1)=0$. Therefore, Problem $\mathbb{P}_{6}$ can be further rewritten as:
\begin{subequations}\label{prob7}
\begin{align}
\mathbb{P}_{7}:&\max\limits_{\bm{x},\bm{\varphi},\bm{\gamma},\bm{T}}  \sum\limits_{n \in \mathcal{N}} \alpha^{(u)}_n [c_n(1 - \varphi_n)d_n - \vartheta_n^{(u)}cost^{(u)}_n] \nonumber \\
&+ \!\!\sum\limits_{n \in \mathcal{N}}\sum\limits_{m \in \mathcal{M}}\alpha^{(s)}_{n,m} (c_{n,m}x_{n,m}\varphi_n d_n - \vartheta_{n,m}^{(s)}cost^{(s)}_{n,m}) \tag{\ref{prob7}}\\
&\text{s.t.} \quad x_{n,m}(x_{n,m}-1)=0, \forall n \in \mathcal{N}, \forall m \in \mathcal{M}, \label{x_constr1_new}\\
&\quad\quad(\text{\ref{x_constr2}})\text{-}(\text{\ref{gamma_constr}}),~(\text{\ref{phi_constr2}}),~(\text{\ref{zeta_constr2}}),~(\text{\ref{Tu_constr}})\text{-}(\text{\ref{Ts_constr}}). \nonumber
\end{align}
\end{subequations}
\begin{theorem}\label{theorem_p7toconvex}
    Problem $\mathbb{P}_{7}$ can be transformed into a solvable convex optimization problem.
\end{theorem}
\begin{proof}
    \textbf{Theorem \ref{theorem_p7toconvex}} is proven by Lemmas \ref{lemma_gamma},
    \ref{lemma_p8top9}, and  \ref{lemma_p9top10} below.
\end{proof}

\begin{lemma}\label{lemma_gamma}
In Problem $\mathbb{P}_{7}$, if focusing on $cost_{n,m}^{(s)}$, $T^{(sp)}_{n,m}$, $T^{(sg)}_{n,m}$ and setting $\omega_b = 1$, we can obtain the optimum solution of $\gamma_{n,m}$ as $\gamma_{n,m}^\star = \frac{1}{2}$.
\end{lemma}
\begin{proof}
    Refer to Appendix \ref{append_lemma_gamma}.
\end{proof}
With \textbf{Lemma \ref{lemma_gamma}}, we obtain the optimum solution of $\gamma_{n,m}$ as $\gamma_{n,m}^\star = \frac{1}{2}$. By substituting the value of $\gamma_{n,m}^\star$, the objective function of Problem $\mathbb{P}_{7}$ becomes:
\begin{talign}
    &\sum_{n \in \mathcal{N}} \alpha^{(u)}_n [c_n(1 - \varphi_n)d_n - \vartheta_n^{(u)}cost^{(u)}_n] \nonumber \\
    &+ \sum_{n \in \mathcal{N}}\sum_{m \in \mathcal{M}}\alpha^{(s)}_{n,m} (c_{n,m}x_{n,m}\varphi_n d_n - \vartheta_{n,m}^{(s)}cost^{(s)}_{n,m})\nonumber \\
    &=\!\!\sum_{n \in \mathcal{N}}\{\!\alpha^{(u)}_n\! c_n d_n \!\!-\!\! \alpha^{(u)}_n \!\vartheta_n^{(u)} \!\omega_t T_n^{(u)}\! \!\!-\!\!\alpha^{(u)}_n \!\vartheta_n^{(u)} \!\omega_e \kappa_n d_n \eta_n f_n^2 \psi_n^2 \nonumber \\
    &- (\alpha^{(u)}_n c_n d_n - \alpha^{(u)}_n \vartheta_n^{(u)} \omega_e \kappa_n d_n \eta_n f_n^2 \psi_n^2) \varphi_n\} \nonumber \\
    &+ \sum_{n \in \mathcal{N}}\sum_{m \in \mathcal{M}}\{\alpha^{(s)}_{n,m}c_{n,m}d_n x_{n,m}\varphi_n - \alpha^{(s)}_{n,m} \vartheta_{n,m}^{(s)} \omega_t T^{(s)}_{n,m} \nonumber \\
    &- \alpha^{(s)}_{n,m} \vartheta_{n,m}^{(s)} \omega_e \frac{\rho_n p_n d_n}{r_{n,m}}x_{n,m}\varphi_n \nonumber \\
    &- \frac{1}{2}\alpha^{(s)}_{n,m} \vartheta_{n,m}^{(s)} \omega_e \kappa_m d_n \eta_m\zeta_{n,m}^2 f_m^2 x_{n,m}\varphi_n\}.
\end{talign}
For brevity, let 
\begin{talign}
    &A_n:=\alpha^{(u)}_n c_n d_n - \alpha^{(u)}_n \vartheta_n^{(u)} \omega_e \kappa_n d_n \eta_n f_n^2 \psi_n^2,\\
    &B_{n,m}:= \frac{1}{2}\alpha^{(s)}_{n,m} \vartheta_{n,m}^{(s)} \omega_e \kappa_m d_n \eta_m\zeta_{n,m}^2 f_m^2 \nonumber \\
    &\quad \quad \quad \quad +\alpha^{(s)}_{n,m} \vartheta_{n,m}^{(s)} \omega_e \frac{\rho_n p_n d_n}{r_{n,m}} - \alpha^{(s)}_{n,m}c_{n,m}d_n,\\
    &C:=\sum_{n \in \mathcal{N}}(-\alpha^{(u)}_n c_n d_n + \alpha^{(u)}_n \vartheta_n^{(u)} \omega_e \kappa_n d_n \eta_n f_n^2 \psi_n^2).
\end{talign}
Let $\bm{A}:=[A_n]|_{n \in \mathcal{N}}$ and $\bm{B}:=[B_{n,m}]|_{n \in \mathcal{N},m \in \mathcal{M}}$. Besides, the ``max'' problem in Problem $\mathbb{P}_{7}$ can also be rewritten as a ``min'' problem:
\begin{subequations}\label{prob8}
\begin{align}
\mathbb{P}_{8}:&\min\limits_{\bm{x},\bm{\varphi},\bm{T}}  C+\sum\limits_{n \in \mathcal{N}} 
\alpha^{(u)}_n \vartheta_n^{(u)} \omega_t T_n^{(u)} + A_n \varphi_n \nonumber \\
&+ \sum_{n \in \mathcal{N}}\sum_{m \in \mathcal{M}}\alpha^{(s)}_{n,m} \vartheta_{n,m}^{(s)} \omega_t T^{(s)}_{n,m} + B_{n,m} x_{n,m} \varphi_n \tag{\ref{prob8}}\\
&\text{s.t.} \quad (\text{\ref{x_constr1_new}}), (\text{\ref{x_constr2}})\text{-}(\text{\ref{varphi_constr}}),
(\text{\ref{phi_constr2}}), (\text{\ref{zeta_constr2}}),
(\text{\ref{Tu_constr}})\text{-}(\text{\ref{Ts_constr}}). \nonumber
\end{align}
\end{subequations}
Since there is the quadratic term $x_{n,m} \varphi_n$ and the quadratic constraint (\text{\ref{x_constr1_new}}), Problem $\mathbb{P}_{8}$ is a quadratically constrained quadratic programming (QCQP) problem. Problem $\mathbb{P}_{8}$ can be addressed via a QCQP-based method, where the quadratic structure is preserved through convex relaxation and subsequently handled with semidefinite programming.
Next, we will explain how to transform Problem $\mathbb{P}_{8}$ into a standard QCQP form problem. We first use a new matrix $\bm{Q}:=(\bm{\varphi}^\intercal,\bm{x_1}^\intercal,\cdots,\bm{x_M}^\intercal)^\intercal$ to combine $\bm{x}$ and $\bm{\varphi}$, where $\bm{\varphi} = (\varphi_1,\cdots,\varphi_N)^\intercal$ and $\bm{x_m} = (x_{1,m},\cdots,x_{N,m})$. Here, we define some auxiliary matrices and vectors to facilitate our transformation. Let 
\begin{talign}
    &e_i:=(0,\cdots,1_{i\text{-th}},\cdots,0)_{NM+N\times1}^\intercal,\\
    &\bm{e}_{i,j}:=(e_i,\cdots,e_j)^\intercal,\\
    &e_{\bar{i}}\!:=\!(0,\cdots,1_{i\text{-th}},\cdots,1_{(i+N)\text{-th}},\cdots,1_{[i+N(M-1)]\text{-th}},\cdots,0)^\intercal,\\
    &\bm{e}_{\bar{i},\bar{j}}:=(e_{\bar{i}}, \cdots, e_{\bar{j}})^\intercal,\\
    &e_{i\rightarrow j}:=(0,\cdots,1_{i\text{-th}}, 1, \cdots, 1_{j\text{-th}}, 0, \cdots,0)^\intercal, i < j,\\
    &\bm{I}_{NM+N\times N}:=(\bm{I}_N,\bm{0}_{N\times NM})^\intercal,\\
    &\bm{I}_{N\rightarrow NM}:=(\bm{I}_N,\cdots,\bm{I}_N)_{N\times NM}.
\end{talign}
To make terms 
\begin{talign}
    &\sum_{n \in \mathcal{N}} \alpha^{(u)}_n \vartheta_n^{(u)} \omega_t T_n^{(u)},\nonumber \\
    &\sum_{n \in \mathcal{N}}\sum_{m \in \mathcal{M}}\alpha^{(s)}_{n,m} \vartheta_{n,m}^{(s)} \omega_t T^{(s)}_{n,m} \nonumber
\end{talign}
succincter, we introduce two variables $T^{(u)}$ and $T^{(s)}$. Then, the constraints (\text{\ref{Tu_constr}}) and (\text{\ref{Ts_constr}}) will be
\begin{talign}
&\sum_{n \in \mathcal{N}} 
\alpha^{(u)}_n \vartheta_n^{(u)} \omega_t T_n^{(u)} \leq T^{(u)}, \label{Tu_constr1}\\
&\sum_{n \in \mathcal{N}}\sum_{m \in \mathcal{M}}\alpha^{(s)}_{n,m} \vartheta_{n,m}^{(s)} \omega_t T^{(s)}_{n,m} \leq T^{(s)},\label{Ts_constr1}
\end{talign}
respectively.
\begin{lemma}\label{lemma_p8top9}
There're matrices $\bm{P}_0$, $\bm{W}_0$, $\bm{P}_2^{(T_u)}$, $\bm{P}_0^{(T_s)}$, and terms $P_1^{(T_u)}$ and $P_1^{(T_s)}$ that can transform Problem $\mathbb{P}_{8}$ into the standard QCQP Problem $\mathbb{P}_{9}$:
\begin{subequations}\label{prob9}
\begin{align}
\mathbb{P}_{9}:&\min\limits_{\bm{Q},T^{(u)},T^{(s)}} \bm{Q}^\intercal \bm{P}_0 \bm{Q} + \bm{W}_0^\intercal \bm{Q} + T^{(u)} + T^{(s)} + C \tag{\ref{prob9}}\\
\text{s.t.} \quad 
&\text{diag}(\boldsymbol{e}_{N+1,NM+N}^\intercal\boldsymbol{Q})(\text{diag}(\boldsymbol{e}_{N+1,NM+N}^\intercal\boldsymbol{Q})-\bm{I})=\bm{0}, \label{x_constr1_qcqp}\\       
& \text{diag}(\boldsymbol{e}_{\overline{1},\overline{M}}^\intercal\boldsymbol{e}_{N+1,NM+N}^\intercal\boldsymbol{Q})=\bm{I}, \label{x_constr2_qcqp}\\
& \text{diag}(\boldsymbol{e}_{1,N}^\intercal\boldsymbol{Q})\preceq\bm{I}, \label{varphi_constr1_qcqp}\\
& \text{diag}(\boldsymbol{e}_{1,N}^\intercal\boldsymbol{Q})\succeq\bm{0}, \label{varphi_constr2_qcqp}\\
&\boldsymbol{\phi}^\intercal\boldsymbol{e}_{N+1,NM+N}^\intercal\boldsymbol{Q}-1 \leq 0, \label{x_phi_constr_qcqp}\\
&\boldsymbol{\zeta}^\intercal\boldsymbol{e}_{N+1,NM+N}^\intercal\boldsymbol{Q}-1 \leq 0, \label{x_zeta_constr_qcqp}\\
& {\bm{P}_2^{(T_u)}}^\intercal \bm{Q} + P_1^{(T_u)} \leq T^{(u)}, \label{Tu_constr_qcqp}\\
& \bm{Q}^\intercal \bm{P}_0^{(T_s)} \bm{Q} + P_1^{(T_s)} \leq T^{(s)}. \label{Ts_constr_qcqp}
\end{align}
\end{subequations}
\end{lemma}
\begin{proof}
    Refer to Appendix \ref{append_lemma_p8top9}.
\end{proof}

Using \textbf{Lemma \ref{lemma_p8top9}}, we can transform Problem $\mathbb{P}_{8}$ into the standard QCQP form Problem $\mathbb{P}_{9}$. However, Problem $\mathbb{P}_{9}$ is still non-convex. Then, we will use the semidefinite programming (SDP) method to transform this QCQP problem into an SDR problem. We introduce a new variable $\bm{S}:=(\bm{Q}^\intercal,1)^\intercal(\bm{Q}^\intercal,1)$. 

\begin{lemma}\label{lemma_p9top10}
There exist matrices $\bm{P}_1$, $\bm{P}_2$, $\bm{P}_3$, $\bm{P}_4$, $\bm{P}_5$, $\bm{P}_6$, $\bm{P}_7$, and $\bm{P}_8$ that can convert the QCQP Problem $\mathbb{P}_{9}$ into the SDR Problem $\mathbb{P}_{10}$:
\begin{subequations}\label{prob10}
\begin{align}
\mathbb{P}_{10}: &\min\limits_{\bm{S},T^{(u)},T^{(s)}}\quad  \text{Tr}(\bm{P}_1 \bm{S})\tag{\ref{prob10}}\\
\text{s.t.} \quad & \text{Tr}(\bm{P}_2 \bm{S})=0, \label{x_constr1_sdr}\\       & \text{Tr}(\bm{P}_3 \bm{S})=0, \label{x_constr2_sdr}\\
         & \text{Tr}(\bm{P}_4 \bm{S})\leq0, \label{varphi_constr_sdr}\\
         & \text{Tr}(\bm{P}_5 \bm{S})\leq0, \label{x_phi_constr_sdr}\\
         & \text{Tr}(\bm{P}_6 \bm{S})\leq0, \label{x_zeta_constr_sdr}\\
         & \text{Tr}(\bm{P}_7 \bm{S})\leq T^{(u)}, \label{Tu_constr_sdr}\\
         & \text{Tr}(\bm{P}_8 \bm{S})\leq T^{(s)}, \label{Ts_constr_sdr}\\
         & \bm{S}\succeq0, \text{rank}(\bm{S})=1, \label{S_constr_sdr}
\end{align}
\end{subequations}
where $Tr(\cdot)$ means the trace of a matrix.
\end{lemma}
\begin{proof}
    Refer to Appendix \ref{append_lemma_p9top10}.
\end{proof}

Based on \textbf{Lemma \ref{lemma_p9top10}}, if we ignore the constraint $\text{rank}(\bm{S})=1$, Problem $\mathbb{P}_{10}$ is finally transformed into a solvable SDR Problem $\mathbb{P}_{10}$. Standard convex solvers can efficiently solve the SDR Problem $\mathbb{P}_{10}$ in polynomial time, providing a continuous version of $\bm{Q}$. However, this version often only serves as the lower bound for the ideal solution and may not satisfy the $\text{rank}(\mathbf{S})=1$ constraint. To rectify this, we apply rounding techniques. The final $NM$ components of $\bm{Q}$, represented by $x_{n,m}$ for every $n \in \mathcal{N}$ and $m \in \mathcal{M}$, reflect the partial connection of users to servers. We label the rounding result of $\bm{x}$ as $\bm{x}^\star$, and the results of $\bm{\varphi}$ in $\bm{Q}$ as $\bm{\varphi}^\star$.

\subsection{Alternating DC technique to solve Problem (\ref{prob10})}
In this section, we introduce an alternative approach, termed the alternating DC technique, to solve Problem $\mathbb{P}_{10}$. Rather than simply dropping the nonconvex rank-one constraint $\text{rank}(\bm{S}) = 1$ and recovering $\bm{S}$ by the Hungarian algorithm in our conference paper \cite{qian2024data}, the alternating DC method provides a more structured way to directly address the nonconvexity.

\textbf{Intuitive explanations to use the alternating DC technique:} Compared to traditional SDR methods, which relax the rank constraint to obtain a convex SDP problem, the alternating DC approach offers several conceptual advantages. In high-dimensional problems, the probability of recovering a rank-one solution from SDR decreases significantly, often requiring additional randomization techniques that lead to performance degradation \cite{jiang2019over}. By contrast, the DC method avoids this relaxation gap by iteratively approximating the original nonconvex feasible set, progressively refining the solution at each step. This enables the alternating DC framework to achieve higher-quality solutions and better scalability without relying on heuristic post-processing.

To illustrate the alternating DC method and prepare for its application to Problem $\mathbb{P}_{10}$, we first present \textbf{Lemma \ref{lemma_DC}}, which provides the theoretical foundation for the proposed approach.
\begin{lemma}\label{lemma_DC}
    If one positive semidefinite matrix $\bm{S} \in \mathbb{C}^{N\times N}$ and $Tr(\bm{S})\leq 1$, we obtain the following equivalent transformation:
    \begin{talign}
        \text{rank}(\bm{S}) = 1 \iff \text{Tr}(\bm{S}) - \lVert \bm{S} \rVert_2 = 0.
    \end{talign}
\end{lemma}
By \textbf{Lemma \ref{lemma_DC}}, we can transform the constraint $\text{rank}(\bm{S}) = 1$ into a penalty term $\varpi \cdot (\text{Tr}(S) - \lVert \bm{S} \rVert_2)$, and put this penalty term into the objective function of Problem $\mathbb{P}_{10}$. $\varpi>0$ is a penalty parameter. Thus, Problem $\mathbb{P}_{10}$ can be transformed into the following new problem:
\begin{subequations}\label{prob11}
\begin{align}
\mathbb{P}_{11}: &\min\limits_{\bm{S},T^{(u)},T^{(s)}}\quad  \text{Tr}(\bm{P}_1 \bm{S}) + \varpi \cdot (\text{Tr}(\bm{S}) - \lVert \bm{S} \rVert_2)\tag{\ref{prob11}}\\
\text{s.t.} \quad & (\text{\ref{x_constr1_sdr}}) - (\text{\ref{Ts_constr_sdr}})\nonumber \\
&\bm{S}\succeq0, \label{S_constr_dc}
\end{align}
\end{subequations}
In the optimization of Problem $\mathbb{P}_{11}$, the penalty term $\varpi \cdot (\text{Tr}(\bm{S}) - \lVert \bm{S} \rVert_2)$ will be forced to approach zero, and the optimal solution $\bm{S}^\star$ obtained satisfies the rank-one constraint. However, because of the existence of the concave term $- \varpi \cdot \lVert \bm{S} \rVert_2$, Problem $\mathbb{P}_{11}$ is still non-convex. To solve it, we need to linearize this non-convex term. Based on the idea of the DC algorithm \cite{jiang2019over}, we present the following transformative optimization:
\begin{subequations}\label{prob12}
\begin{align}
\mathbb{P}_{12}: &\min\limits_{\bm{S},T^{(u)},T^{(s)}}\quad  \text{Tr}(\bm{P}_1 \bm{S}) + \varpi \cdot \langle \bm{I}- \partial\lVert \bm{S}^{(c)} \rVert_2 , \bm{S}\rangle\tag{\ref{prob12}}\\
\text{s.t.} \quad & (\text{\ref{x_constr1_sdr}}) - (\text{\ref{Ts_constr_sdr}}), (\text{\ref{S_constr_dc}}),\nonumber
\end{align}
\end{subequations}
where $\bm{S}^{(c)}$ is the value of $\bm{S}$ at the last iteration, and $\langle\,,\rangle$ denotes the inner product operation. The subgradient $\partial\lVert \bm{S}^{(c)} \rVert_2$ can be obtained as $\bm{s}_e \bm{s}_e^\intercal$, where $\bm{s}_e$ is the leading eigenvector of $\bm{S}^{(c)}$. Now, Problem $\mathbb{P}_{12}$ is a convex optimization, which common convex tools can solve. A critical point solution is guaranteed \cite{jiang2019over}.

\subsection{Performance comparison of different rounding techniques}
In this section, we study what is the most efficient rounding techniques to recover the continuous $\bm{x}$ to a binary discrete value. 
\subsubsection{Rounding techniques}
Here, we present the following rounding techniques:
\paragraph{Hungarian algorithm}
The Hungarian (Kuhn-Munkres) algorithm is a classic combinatorial optimization method that solves the assignment problem in polynomial time. By turning the continuous solution into a cost matrix, the Hungarian algorithm can be applied to obtain a binary assignment that aligns closely with the original continuous solution. Here, we use the Hungarian algorithm \cite{kuhn1955hungarian}, augmented with zero vectors, is used to identify the optimal matching, denoted as $\mathcal{X}$. Within this matching, we set $x_{n,m}$ to 1 if nodes $n$ and $m$ are paired and 0 otherwise.
\paragraph{Randomized rounding}
Randomized rounding converts the continuous solution to a binary solution probabilistically. For each element $x_{n,m} \in [0,1]$, interpret it as a probability:
\begin{talign}
    x_{n,m} = \begin{cases}
        1, & \text{with probability $x_{n,m}$}  \\
        0, & \text{with probability $1 - x_{n,m}$}.
    \end{cases}
\end{talign}
At the same time, we ensure that $\sum_{m\in\mathcal{M}}x_{n,m}=1$.
\paragraph{Solve a secondary discrete problem}
Once the SDP solution $\bm{x}$ is obtained, we can formulate a new optimization problem that approximates the binary solution. The new optimization problem is shown as follows:
\begin{subequations}\label{prob13}
\begin{align}
\mathbb{P}_{13}: &\min\limits_{\bm{x}}\quad \lVert \bm{x} - \bm{x}^{(SDP)} \rVert_2^2 \tag{\ref{prob13}}\\
\text{s.t.} \quad & x_{n,m} \in \{0, 1\},\nonumber
\end{align}
\end{subequations}
where $\bm{x}^{(SDP)}$ is the obtained continuous SDP solution. This optimization problem is a simple mixed integer programming, which can be solved by Mosek in MATLAB.
\paragraph{Rank-1 approximation}
In the rank-1 approximation method, the continuous matrix solution is projected to a \hbox{rank-1} binary matrix solution. We first perform the eigenvalue decomposition of the obtained continuous $\bm{x}$. Take the largest eigenvector, which is associated with the dominant eigenvalue, and then round the entries of the eigenvector to binary values.
\paragraph{Greedy rounding}
In the greedy rounding, we set $x_{n,m} = 1$ with the largest continuous value at each step.
\subsubsection{Numerical results}
In this section, we present numerical results of different rounding techniques in Fig. \ref{fig.rounding_compare}. The default setting can be seen in Section \ref{sec.simulation_results}. The continuous solution $\bm{x}$ is obtained by solving Problem $\mathbb{P}_{10}$. We conduct 500 repeat simulations, and the initial setting of $\bm{x}$ and $\bm{\phi}$ are randomly chosen within their constraints.

Among the five rounding methods, Rank-1 approximation emerges as the most suitable due to its superior $\mathbb{P}_{1}$ objective function value, which is the highest among all, coupled with a low and consistent running time. While the Hungarian method and Greedy rounding are computationally efficient, they provide suboptimal $\mathbb{P}_{1}$ objective function values because they primarily treat the relaxed solution as a static cost matrix, without fully leveraging the structural information embedded in the continuous relaxation.

The secondary discrete problem method achieves decent objective function values but suffers from significantly higher computational time, making it impractical for time-sensitive applications. Randomized rounding is fast and performs well but is less consistent and introduces randomness, which can lead to variability in performance.

In contrast, the Rank-1 approximation method not only shows empirical superiority but is also theoretically well-aligned with the problem's structure. Since the original optimization problem inherently exhibits a low-rank nature, where the ideal assignment matrix should be rank-one, the principal eigenvector extracted during rank-1 approximation effectively captures the dominant decision direction in the relaxed solution \cite{luo2010semidefinite}. By rounding based on this principal component, the method preserves the essential coupling among variables and maintains closer adherence to the original nonconvex feasible set. Overall, Rank-1 approximation strikes the best balance between solution quality and computational efficiency, making it the preferred choice. Thus, we choose the Rank-1 approximation method as the rounding technique in this paper.
\begin{figure}[t] 
\subfigure[Performance comparison.]{\includegraphics[width=.24\textwidth]{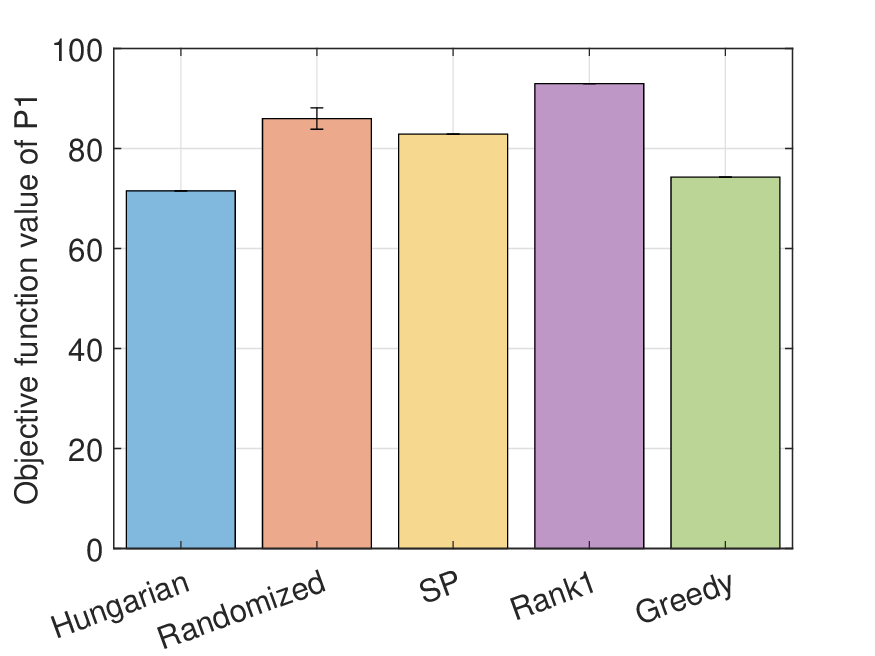}\label{fig.rounding_value}}
\subfigure[Running time comparison.]{\includegraphics[width=.24\textwidth]{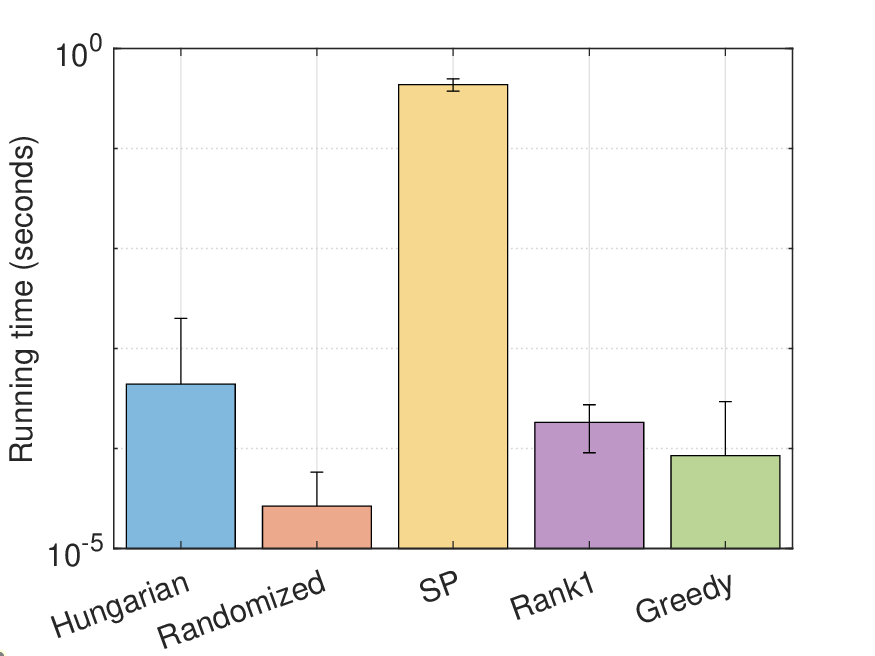}\label{fig.rounding_time}} 
\caption{Performance comparison of different rounding techniques. ``Hungarian'', ``randomized'', ``SP'', ``Rank1'', and ``Greedy'' denote the Hungarian algorithm, randomized rounding, secondary discrete problem, rank-1 approximation, and greedy rounding, respectively.} 
\label{fig.rounding_compare}
\end{figure}
\begin{figure}[t] 
\subfigure[Performance comparison of different penalty weights.]{\includegraphics[width=.24\textwidth]{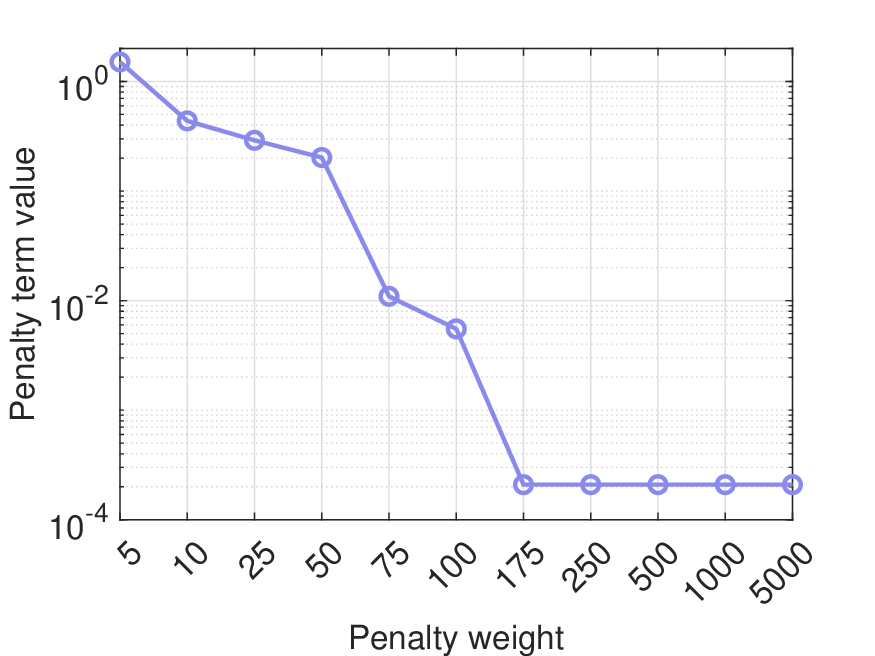}\label{fig.comparison_varpi}}
\subfigure[Performance with and without dropping the rank-one constraint.]{\includegraphics[width=.24\textwidth]{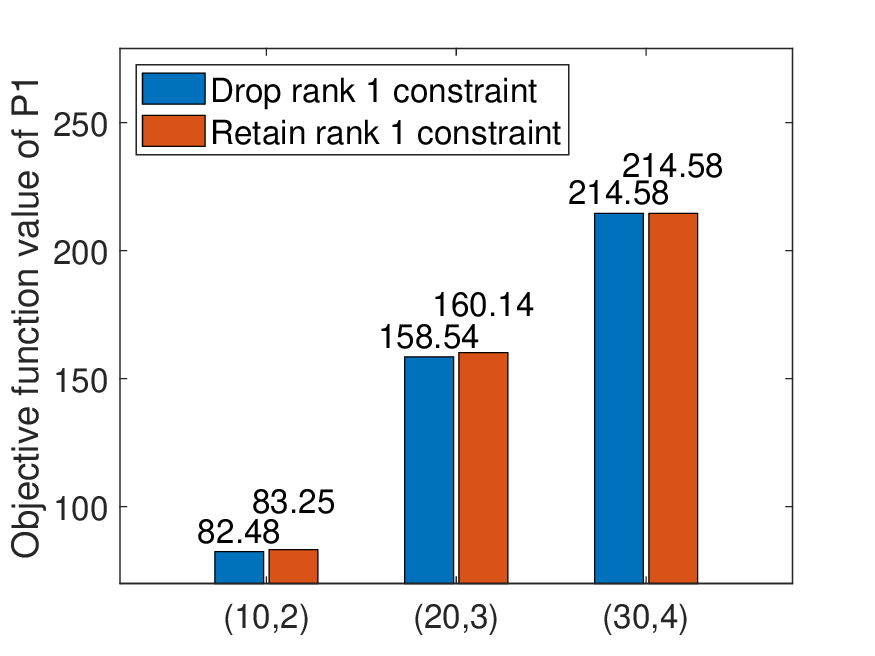}\label{fig.comparison_sdp}} 
\caption{Performance comparison of different penalty weights and SDP methods.} 
\label{fig.sdp_compare}
\end{figure}
\begin{algorithm}[htbp]
\caption{Proposed DAUR algorithm.}
\label{algorithm}

Initialize $i \leftarrow -1$ and for all $n \in \mathcal{N}, m \in \mathcal{M}$: $\bm{x}^{(0)} = (\bm{e_1},\cdots,\bm{e_M})^\intercal$, $\varphi_n^{(0)}=0.5$, $\phi_{n,m}^{(0)}=\frac{1}{N}$, $\rho_n^{(0)}=1$, $\zeta_{n,m}^{(0)}=\frac{1}{N}$, $\psi_n^{(0)}=1$, $\gamma_{n,m} = 0.5$;

Calculate $\bm{\alpha}^{(0)}, \bm{\vartheta}^{(0)}$ with $\bm{x}^{(0)} = (\bm{e_1},\cdots,\bm{e_M})^\intercal$, $\varphi_n^{(0)}=0.5$, $\phi_{n,m}^{(0)}=\frac{1}{N}$, $\rho_n^{(0)}=1$, $\zeta_{n,m}^{(0)}=\frac{1}{N}$, $\psi_n^{(0)}=1$, $\gamma_{n,m} = 0.5$;

\Repeat{$\frac{V_{\mathbb{P}_{3}}(\bm{x}^{(i+1)}, \bm{\varphi}^{(i+1)}, \bm{\phi}^{(i+1)}, \bm{\rho}^{(i+1)}, \bm{\zeta}^{(i+1)}, \bm{\psi}^{(i+1)})}{V_{\mathbb{P}_{3}}(\bm{x}^{(i)}, \bm{\varphi}^{(i)}, \bm{\phi}^{(i)}, \bm{\rho}^{(i)}, \bm{\zeta}^{(i)}, \bm{\psi}^{(i)})}- 1 \leq \epsilon_3$, where $\epsilon_3$ is a small positive number}{
Let $i\leftarrow i+1$;

Initialize $j = -1$; 

Calculate $\bm{\upsilon}^{(i,0)}$ with $\bm{x}^{(i)}, \bm{\varphi}^{(i)}, \bm{\phi}^{(i)}, \bm{\rho}^{(i)}, \bm{\zeta}^{(i)}, \bm{\psi}^{(i)}$;

Set $[\bm{\phi}^{(i,0)}, \bm{\rho}^{(i,0)}, \bm{\zeta}^{(i,0)}, \bm{\psi}^{(i,0)}, \bm{T}^{(i,0)}]\leftarrow$ $[\bm{\phi}^{(i)}, \bm{\rho}^{(i)}, \bm{\zeta}^{(i)}, \bm{\psi}^{(i)}, \bm{T}^{(i)}]$;

\Repeat{$\frac{V_{\mathbb{P}_{5}}(\bm{\phi}^{(i,j+1)}, \bm{\rho}^{(i,j+1)}, \bm{\zeta}^{(i,j+1)}, \bm{\psi}^{(i,j+1)})}{V_{\mathbb{P}_{5}}(\bm{\phi}^{(i,j)}, \bm{\rho}^{(i,j)}, \bm{\zeta}^{(i,j)}, \bm{\psi}^{(i,j)})}- 1 \leq \epsilon_1$, where $\epsilon_1$ is a small positive number}{
Let $j\leftarrow j+1$;

Obtain $[\bm{\phi}^{(i,j+1)}, \bm{\rho}^{(i,j+1)}, \bm{\zeta}^{(i,j+1)}, \bm{\psi}^{(i,j+1)}, \bm{T}^{(i,j+1)}]$ by solving Problem $\mathbb{P}_{5}$ with $\bm{\upsilon}^{(i,j)}$;

Update $\bm{\upsilon}^{(i,j+1)}$ with $[\bm{\phi}^{(i,j+1)}, \bm{\rho}^{(i,j+1)}, \bm{\zeta}^{(i,j+1)}, \bm{\psi}^{(i,j+1)}, \bm{T}^{(i,j+1)}]$;
}
Return $[\bm{\phi}^{(i,j+1)}, \bm{\rho}^{(i,j+1)}, \bm{\zeta}^{(i,j+1)}, \bm{\psi}^{(i,j+1)}]$ as a solution to Problem $\mathbb{P}_{5}$;

Set $[\bm{\phi}^{(i+1)}, \bm{\rho}^{(i+1)}, \bm{\zeta}^{(i+1)}, \bm{\psi}^{(i+1)}]$ $\leftarrow$ $[\bm{\phi}^{(i,j+1)}, \bm{\rho}^{(i,j+1)}, \bm{\zeta}^{(i,j+1)}, \bm{\psi}^{(i,j+1)}]$;

Initialize $j = -1$;

Set $[\bm{x}^{(i,0)}, \bm{\varphi}^{(i,0)}] \leftarrow [\bm{x}^{(i)}, \bm{\varphi}^{(i)}]$;

Initialize $[\bm{T}^{(i,0)}, \bm{A}^{(i,0)}, \bm{B}^{(i,0)}, C^{(i,0)}, \bm{P}_k^{(i,0)}]\leftarrow$ $[\bm{T}^{(i)}, \bm{A}^{(i)}, \bm{B}^{(i)}, C^{(i)}, \bm{P}_k^{(i)}]$, where $k \in \{1,2,\cdots,8\}$;

\Repeat{$\frac{V_{\mathbb{P}_{12}}(\bm{x}^{(i,j+1)}, \bm{\varphi}^{(i,j+1)}}{V_{\mathbb{P}_{12}}(\bm{x}^{(i,j)}, \bm{\varphi}^{(i,j)})}- 1 \leq \epsilon_2$, where $\epsilon_2$ is a small positive number}{
Let $j\leftarrow j+1$;

Obtain $[\bm{x}^{(i,j+1)}, \bm{\varphi}^{(i,j+1)}]$ of continuous values by solving Problem $\mathbb{P}_{12}$;

Update $[\bm{T}^{(i,j+1)}, \bm{A}^{(i,j+1)}, \bm{B}^{(i,j+1)}, C^{(i,j+1)}, \bm{P}_k^{(i,j+1)}$ with $[\bm{x}^{(i,j+1)}, \bm{\varphi}^{(i,j+1)}]$;
}

Return $[\bm{x}^{(i,j+1)}, \bm{\varphi}^{(i,j+1)}]$ as a solution to the SDR Problem $\mathbb{P}_{12}$;

Use the rank-1 approximation method to obtain the discrete association results as $\bm{x}_\star^{(i,j+1)}$.

Set $[\bm{x}^{(i+1)}, \bm{\varphi}^{(i+1)}]\leftarrow$ $[\bm{x}_\star^{(i,j+1)}, \bm{\varphi}^{(i,j+1)}]$;

Update $[\bm{\alpha}^{(i+1)}, \bm{\vartheta}^{(i+1)}]$ with $[\bm{x}^{(i+1)}, \bm{\varphi}^{(i+1)}, \bm{\phi}^{(i+1)}, \bm{\rho}^{(i+1)}, \bm{\zeta}^{(i+1)}, \bm{\psi}^{(i+1)}]$;
}
Return $[\bm{x}^{(i+1)}, \bm{\varphi}^{(i+1)}, \bm{\phi}^{(i+1)}, \bm{\rho}^{(i+1)}, \bm{\zeta}^{(i+1)}, \bm{\psi}^{(i+1)}]$ as a solution $[\bm{x}^\star, \bm{\varphi}^\star, \bm{\phi}^\star, \bm{\rho}^\star, \bm{\zeta}^\star, \bm{\psi}^\star]$ to Problem $\mathbb{P}_{3}$.
\end{algorithm}
\subsection{Performance comparison between optimization with and without dropping the rank-1 constraint}
In this section, we study which method (drop the rank-1 constraint or not) is better to solve Problem $\mathbb{P}_{10}$. 

First, we try to find a suitable penalty weight parameter $\varpi$ to perform optimization of Problem $\mathbb{P}_{12}$. In Fig. \ref{fig.comparison_varpi}, the final penalty term values after solving Problem $\mathbb{P}_{12}$ of different penalty weights, i.e., $\varpi$ are given. The value of the penalty term on the Y-axis of Fig. \ref{fig.comparison_varpi} is the value of $(\text{Tr}(\bm{S}) - \lVert \bm{S} \rVert_2)$. From this figure, it is known that a larger $\varpi$ would be better to force the penalty term to approach zero. However, when $\varpi$ is greater than 175, the order of the penalty term value will not change much. The larger numbers of penalty weights do not significantly affect the value of the penalty term. Therefore, we set the penalty weight $\varpi$ as 175.

In Fig. \ref{fig.comparison_sdp}, we compare the objective function value of Problem $\mathbb{P}_{1}$ with the solutions to Problem $\mathbb{P}_{10}$ by dropping or retaining the rank-1 constraint (i.e., $\text{rank} (\bm{S}) = 1$). Three use-server configurations are chosen as $(N, M) \in \{(10, 2), (20, 3), (30, 4)\}$. The results in Fig. \ref{fig.comparison_sdp} indicate that retaining the rank-1 constraint provides a slight improvement over dropping it. The reason for those results is that dropping the rank-1 constraint may lead to higher-rank or less structured results, which also introduces instability or numerical issues. Thus, in this paper, retaining the rank-1 constraint method is chosen to solve the SDP problem.

\subsection{The whole algorithm procedure}
\begin{figure}[t]
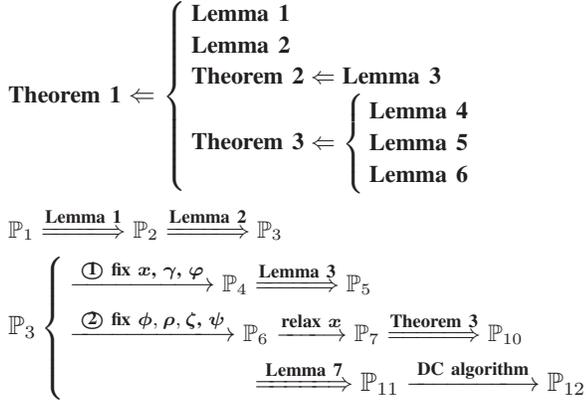

\begin{talign}
    &\small\textnormal{\textbf{Theorem \ref{theorem_solvep1}}} \Leftarrow 
    \begin{cases}
        \textnormal{\textbf{Lemma \ref{lemma_p1top2}}}\\
        \textnormal{\textbf{Lemma \ref{lemma_p2top3}}}\\
        \textnormal{\textbf{Theorem \ref{theorem_p4toconcave}}} \Leftarrow \textnormal{\textbf{Lemma \ref{lemma_p4top5}}}\\
        \textnormal{\textbf{Theorem \ref{theorem_p7toconvex}}} \Leftarrow 
        \begin{cases}
            \textnormal{\textbf{Lemma \ref{lemma_gamma}}}\\
            \textnormal{\textbf{Lemma \ref{lemma_p8top9}}}\\
            \textnormal{\textbf{Lemma \ref{lemma_p9top10}}} \nonumber
        \end{cases}
    \end{cases}\\
    &\small\textnormal{\text{$\mathbb{P}_{1}$}} \xLongrightarrow{\small\stackrel~{\textnormal{\textbf{Lemma \ref{lemma_p1top2}}}}} \textnormal{\text{$\mathbb{P}_{2}$}} \xLongrightarrow{\small\stackrel~{\textnormal{\textbf{Lemma \ref{lemma_p2top3}}}}} \textnormal{\text{$\mathbb{P}_{3}$}}\nonumber \\
    &\textnormal{\text{$\mathbb{P}_{3}$}}
    \begin{cases}
    \xlongrightarrow{\small\stackrel~{\textnormal{\textbf{\textcircled{1}~fix $\bm{x}$, $\bm{\gamma}$, $\bm{\varphi}$}}}}\small\textnormal{\text{$\mathbb{P}_{4}$}} \xLongrightarrow{\small\stackrel~{\textnormal{\textbf{Lemma \ref{lemma_p4top5}}}}} \textnormal{\text{$\mathbb{P}_{5}$}}\nonumber\\
    \xlongrightarrow{\small\stackrel~{\textnormal{\textbf{\textcircled{2}~fix $\bm{\phi},\bm{\rho},\bm{\zeta}$, $\bm{\psi}$}}}}\small\textnormal{\text{$\mathbb{P}_{6}$}} \xlongrightarrow{\small\stackrel~{\textnormal{\textbf{relax $\bm{x}$}}}} \textnormal{\text{$\mathbb{P}_{7}$}} \xLongrightarrow{\small\stackrel~{\textnormal{\textbf{Theorem \ref{theorem_p7toconvex}}}}}\textnormal{\text{$\mathbb{P}_{10}$}}\nonumber\\
    \hspace{70pt}\xLongrightarrow{\small\stackrel~{\textnormal{\textbf{Lemma \ref{lemma_DC}}}}} \textnormal{\text{$\mathbb{P}_{11}$}}\xlongrightarrow{\small\stackrel~{\textnormal{\textbf{DC algorithm}}}} \textnormal{\text{$\mathbb{P}_{12}$}}\nonumber
    \end{cases}
\end{talign}
    \caption{Relationships between Problems, Theorems, and Lemmas, where ``A $\Leftarrow$ B'' means that A is implied by B, and ``A $\Rightarrow$ B'' means that A implies B.}
    \label{fig:relat_theorem_lemma}
\end{figure}
Let the objective function value of Problem $\mathbb{P}_{i}$ be $V_{\mathbb{P}_{i}}$. We present the whole procedure of the proposed DAUR algorithm in Algorithm \ref{algorithm} and the relationships between Problems, Theorems, and Lemmas in \mbox{Fig. \ref{fig:relat_theorem_lemma}}. The core theorem, \textbf{Theorem~\ref{theorem_solvep1}}, is supported by a series of Lemmas and Theorems:
\begin{itemize}
    \item Starting from $\mathbb{P}_{1}$, we first reformulate it into $\mathbb{P}_{2}$ and then into $\mathbb{P}_{3}$, through \textbf{Lemma~\ref{lemma_p1top2}} and \textbf{Lemma~\ref{lemma_p2top3}}, respectively.
    \item The problem $\mathbb{P}_{3}$ is further branched into two reformulation pathways:
        \begin{enumerate}
            \item Fixing $\bm{x}$, $\bm{\gamma}$, and $\bm{\varphi}$, we obtain $\mathbb{P}_{4}$, which is transformed to $\mathbb{P}_{5}$ via \textbf{Lemma~\ref{lemma_p4top5}}, and further transformed to a concave form through \textbf{Theorem~\ref{theorem_p4toconcave}}.
            \item Fixing $\bm{\phi},\bm{\rho},\bm{\zeta}$, and $\bm{\psi}$, the problem becomes $\mathbb{P}_{6}$, and is relaxed to $\mathbb{P}_{7}$, then to $\mathbb{P}_{10}$ using \textbf{Theorem~\ref{theorem_p7toconvex}}, which itself depends on \textbf{Lemma~\ref{lemma_gamma}}, \textbf{Lemma~\ref{lemma_p8top9}}, and \textbf{Lemma~\ref{lemma_p9top10}}.
            \end{enumerate}
    \item Further, $\mathbb{P}_{10}$ is transformed into $\mathbb{P}_{11}$ using \textbf{Lemma~\ref{lemma_DC}}, and solved via a DC algorithm, resulting in the final form $\mathbb{P}_{12}$.
\end{itemize}

Here we summarize the overall flow of the optimization algorithm. At the $i$-th iteration, we first initialize $\bm{\alpha}^{(i-1)}, \bm{\vartheta}^{(i-1)}$ with $\bm{x}^{(i-1)}$, $\bm{\varphi}^{(i-1)}$, $\bm{\phi}^{(i-1)}$, $\bm{\rho}^{(i-1)}$, $\bm{\zeta}^{(i-1)}$, $\bm{\psi}^{(i-1)}$. Then, we fix $\bm{\alpha}, \bm{\vartheta}$ as $\bm{\alpha}^{(i-1)}, \bm{\vartheta}^{(i-1)}$ and optimize $\bm{x}$, $\bm{\varphi}$, $\bm{\phi}$, $\bm{\rho}$, $\bm{\zeta}$, $\bm{\psi}$. For the optimization of $\bm{x}$, $\bm{\varphi}$, $\bm{\phi}$, $\bm{\rho}$, $\bm{\zeta}$, $\bm{\psi}$, we use the alternative optimization technique.

In the first step, we fix $\bm{x}$, $\bm{\varphi}$ as $\bm{x}^{(i-1)}$, $\bm{\varphi}^{(i-1)}$ and optimize $\bm{\phi}$, $\bm{\rho}$, $\bm{\zeta}$, $\bm{\psi}$. At this optimization step, we also introduce an auxiliary variable $\upsilon^{(s)}_{n,m}$, where $\upsilon^{(s)}_{n,m} = \frac{1}{2x_{n,m}\rho_np_n\varphi_nd_nr_{n,m}}$ to transform Problem $\mathbb{P}_{4}$ into a solvable concave problem $\mathbb{P}_{5}$. At the $j$-th inner iteration, we initialize $\bm{\upsilon}^{\bm{(s)}(i-1,j-1)}$ with $\bm{x}^{(i-1)}$, $\bm{\varphi}^{(i-1)}$, $\bm{\phi}^{(i-1,j-1)}$, $\bm{\rho}^{(i-1,j-1)}$, $\bm{\zeta}^{(i-1,j-1)}$, $\bm{\psi}^{(i-1,j-1)}$. We fix $\bm{\upsilon}^{\bm{(s)}}$ as $\bm{\upsilon}^{\bm{(s)}(i-1,j-1)}$ and optimize $\bm{\phi}$, $\bm{\rho}$, $\bm{\zeta}$, $\bm{\psi}$. Then we obtain the optimization results $\bm{\phi}^{(i-1,j)}$, $\bm{\rho}^{(i-1,j)}$, $\bm{\zeta}^{(i-1,j)}$, $\bm{\psi}^{(i-1,j)}$ and update $\bm{\upsilon}^{\bm{(s)}(i-1,j)}$ with these results. This optimization cycle is repeated until the difference in the objective function value of Problem $\mathbb{P}_{5}$ between the $j$-th and $(j-1)$-th iterations falls below a predefined threshold. We set the results of this alternative optimization step as $\bm{\phi}^{(i)}$, $\bm{\rho}^{(i)}$, $\bm{\zeta}^{(i)}$, $\bm{\psi}^{(i)}$.

In the second step, we fix the $\bm{\phi}$, $\bm{\rho}$, $\bm{\zeta}$, $\bm{\psi}$ as $\bm{\phi}^{(i)}$, $\bm{\rho}^{(i)}$, $\bm{\zeta}^{(i)}$, $\bm{\psi}^{(i)}$ and optimize $\bm{x}$, $\bm{\varphi}$. Then we first obtain $\bm{\varphi}^{(i)}$ and the continuous solution of $\bm{x}$ by solving Problem $\mathbb{P}_{12}$. Next, we use the rank-1 approximation method to obtain the discrete solution of $\bm{x}$ and denote it as $\bm{x}^{(i)}$.  Until now, we have obtained $\bm{x}^{(i)}$, $\bm{\varphi}^{(i)}$, $\bm{\phi}^{(i)}$, $\bm{\rho}^{(i)}$, $\bm{\zeta}^{(i)}$, $\bm{\psi}^{(i)}$. Update $\bm{\alpha}^{(i)}, \bm{\vartheta}^{(i)}$ with those results.

Repeat these two optimization steps until the difference in the objective function value of Problem $\mathbb{P}_{3}$ between the $i$-th and \mbox{$(i-1)$-th}  iterations falls in a predefined threshold. Then, we set the optimization results as $\bm{x}^\star$, $\bm{\varphi}^\star$, $\bm{\phi}^\star$, $\bm{\rho}^\star$, $\bm{\zeta}^\star$, $\bm{\psi}^\star$.

\begin{figure*}[!htbp] 
\subfigure[Convergence of the FP method.]{\includegraphics[width=.33\textwidth]{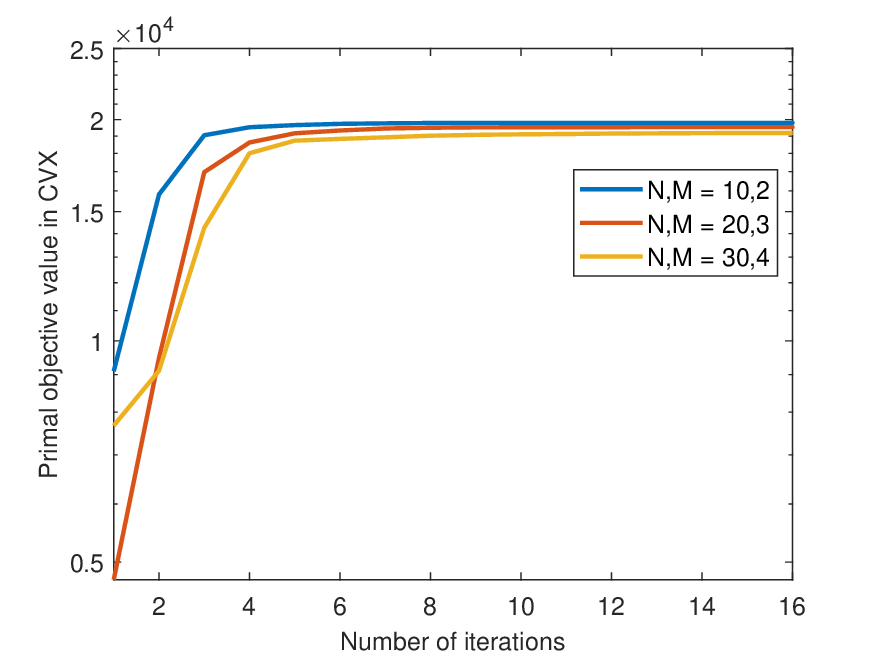}\label{fig.convergence of FP}}
\subfigure[Convergence of the QCQP method.]{\includegraphics[width=.33\textwidth]{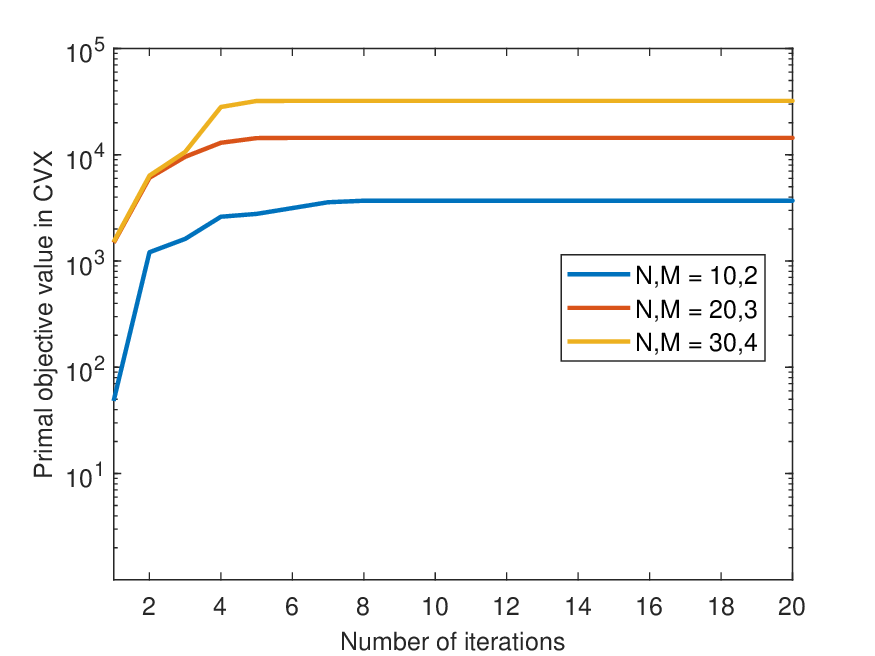}\label{fig.convergence of QCQP}}
\subfigure[Performance comparison between the DAUR algorithm and other baselines.]{\includegraphics[width=.33\textwidth]{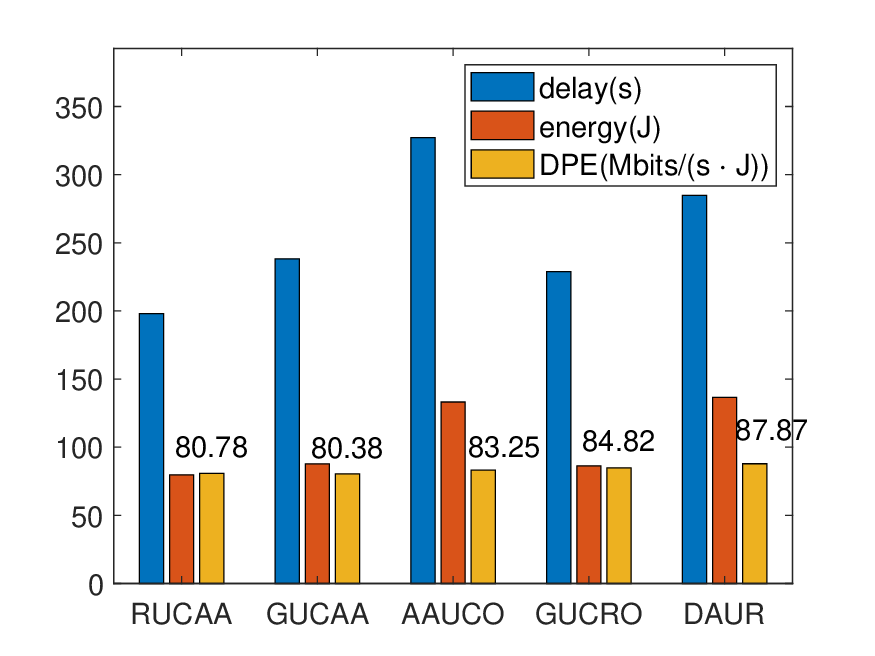}\label{fig.DAUR performance}}  
\caption{Convergence of FP and QCQP methods; Performance comparison of the DAUR algorithm with baselines.} 
\label{fig.convergence_qcqp_fp}
\end{figure*}
\begin{figure}[t] 
\subfigure[Convergence over total rounds.]{\includegraphics[width=.24\textwidth]{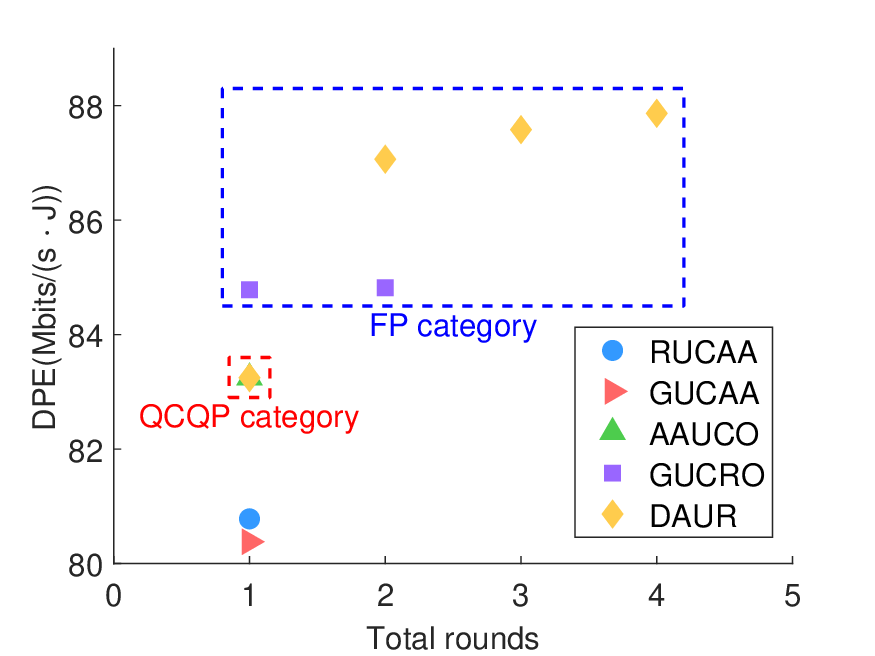}\label{fig.convergence_iteration}}
\subfigure[Convergence over running time.]{\includegraphics[width=.24\textwidth]{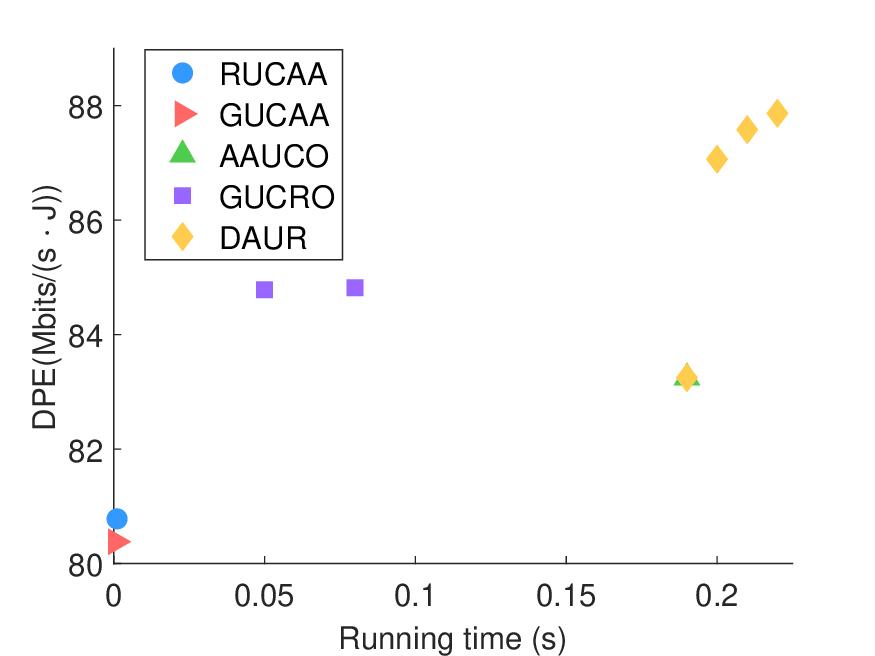}\label{fig.convergence_time}} 
\caption{Comparison of convergence behavior with baselines.} 
\label{fig.convergence_compare}
\end{figure}
\subsection{Novelty of our proposed DAUR algorithm}
In this paper, we address maximizing the combined DPE of users and servers in a blockchain-enhanced Metaverse wireless system using the DAUR algorithm. This algorithm optimizes user-server associations, work offloading ratios, and task-specific computing resource distribution ratios together, as well as jointly optimizes communication and computational resources like bandwidth, transmit power, and computing allocations for both users and servers. \mbox{Unlike} previous methods that treat communication and computational resources separately~\cite{deng2022blockchain,dai2018joint}, our approach integrates them into a unified optimization problem, leading to better solutions than traditional alternating optimization methods. Additionally, the DAUR algorithm's application extends beyond DPE maximization; it's also suitable for solving energy efficiency and various utility-cost problems. For non-concave utility functions, we use successive convex approximation (SCA) to enable DAUR's application in mobile edge computing for user connection and resource allocation.

\subsection{Discussion of practical application cases:}
The proposed DAUR algorithm can be applied in a variety of blockchain-empowered Metaverse scenarios to enhance system efficiency and user experience. For instance, in NFT marketplaces, users frequently engage in the creation, trading, and verification of non-fungible tokens, which require significant computational resources for cryptographic operations and metadata processing. By dynamically optimizing user association and resource allocation, the DAUR algorithm significantly reduces the processing delay and energy consumption, and improves the data processing efficiency, thereby accelerating NFT transactions and extending device battery life.

Similarly, during large-scale virtual events, e.g., concerts, exhibitions, and conferences within the Metaverse, massive volumes of real-time multimedia data must be processed, transmitted, and synchronized across decentralized blockchain nodes. In these scenarios, DAUR enables efficient offloading and load balancing by adapting to dynamic network conditions and user densities, ensuring seamless real-time interactions with minimal latency.

In future work, we aim to further validate the algorithm's effectiveness through real-world system deployments and application-specific case studies.

\section{Complexity Analysis}\label{sec.complexity_analysis}
In this section, the computational complexity of the proposed DAUR algorithm is analyzed by examining its main components. In the optimization phase from Line 8 to Line 12 of Algorithm 1, the problem involves $3N+3NM$ variables and $3N+3NM+2M$ constraints, resulting in a worst-case complexity of $\mathcal{O}((N^{3.5}+M^{3.5}+N^{3.5}M^{3.5})\log(\frac{1}{\epsilon_1}))$ with a given solution accuracy $\epsilon_1 > 0$~\cite{dai2018joint}. Similarly, the optimization phase from Line 18 to Line 22, with $N+NM$ variables and $NM+2N+2M+2$ constraints, also has a worst-case complexity of $\mathcal{O}((N^{3.5}+M^{3.5}+N^{3.5}M^{3.5})\log(\frac{1}{\epsilon_2}))$ with a given solution accuracy $\epsilon_2 > 0$. Additionally, the rank-1 approximation step using singular value decomposition (SVD) incurs a cost of $\mathcal{O}(N^3M^3)$ \cite{li2019tutorial}. Assuming the algorithm runs for $\mathcal{I}$ outer iterations, the total complexity is \mbox{$\mathcal{I}\times\mathcal{O}((N^{3.5}+M^{3.5}+N^{3.5}M^{3.5})\log(\frac{1}{\epsilon_3}))$} with a global solution accuracy $\epsilon_3 > 0$. In simulations, the number of rounds for the outer iteration, the QCQP method, and the FP method are relatively small (the numbers are often in single digits, see Table \ref{tab:DAURConvergence} on the next page). The optimization problem can be efficiently solved by convex solvers, making the proposed DAUR algorithm computationally feasible for practical applications.
\begin{figure*}[t] 
\subfigure[Performance comparisons of different total bandwidth.]{\includegraphics[width=.33\textwidth]{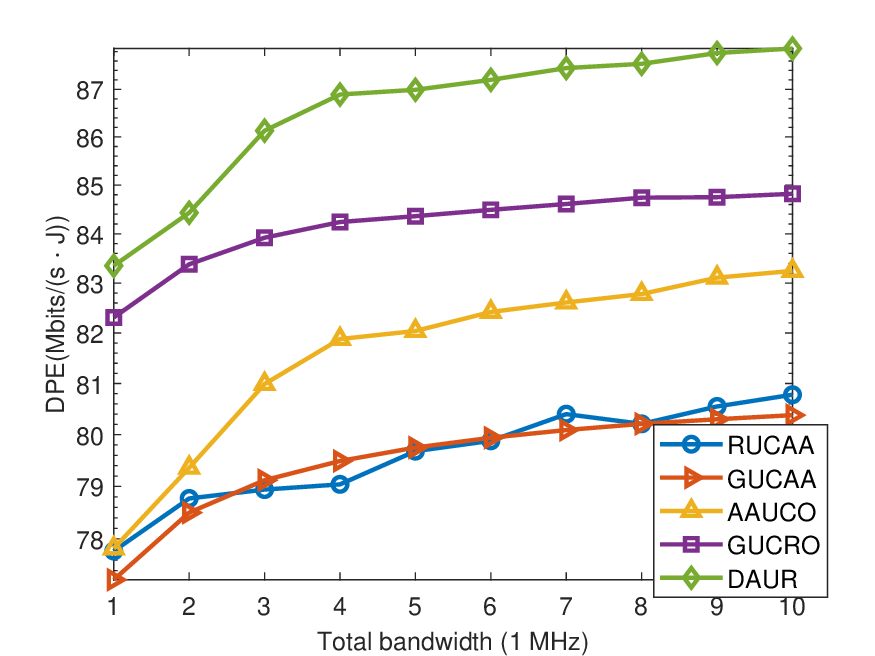}\label{fig.comparison bandwidth}}
\subfigure[Performance comparisons of different server frequency.]{\includegraphics[width=.33\textwidth]{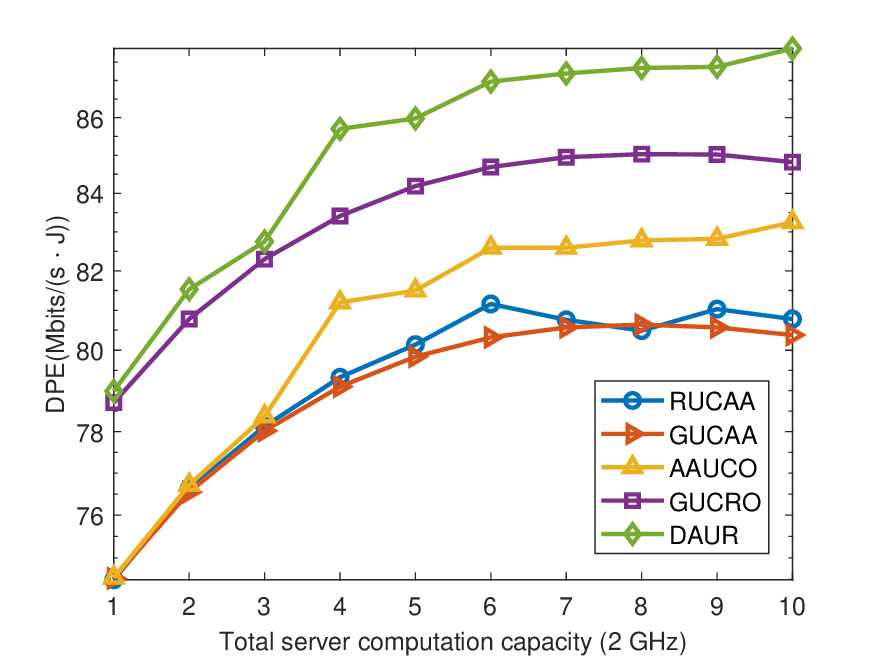}\label{fig.comparison server frequency}}
\subfigure[Performance comparisons of different user frequency.]{\includegraphics[width=.33\textwidth]{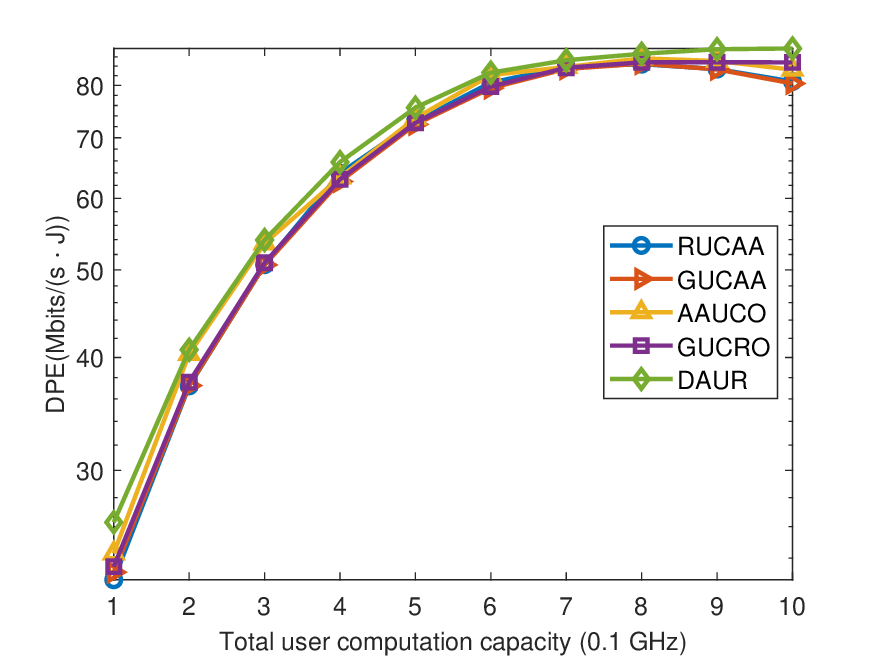}\label{fig.comparison user frequency}}
\subfigure[Performance comparisons of different user transmit power.]{\includegraphics[width=.33\textwidth]{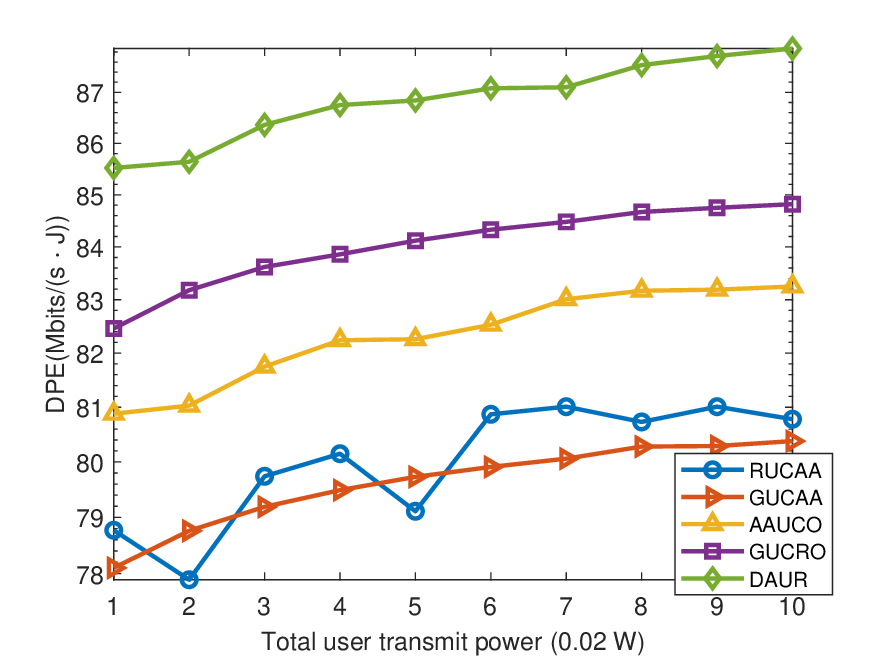}\label{fig.comparison user transmit power}} 
\subfigure[Performance comparisons of different weight ratios.]{\includegraphics[width=.33\textwidth]{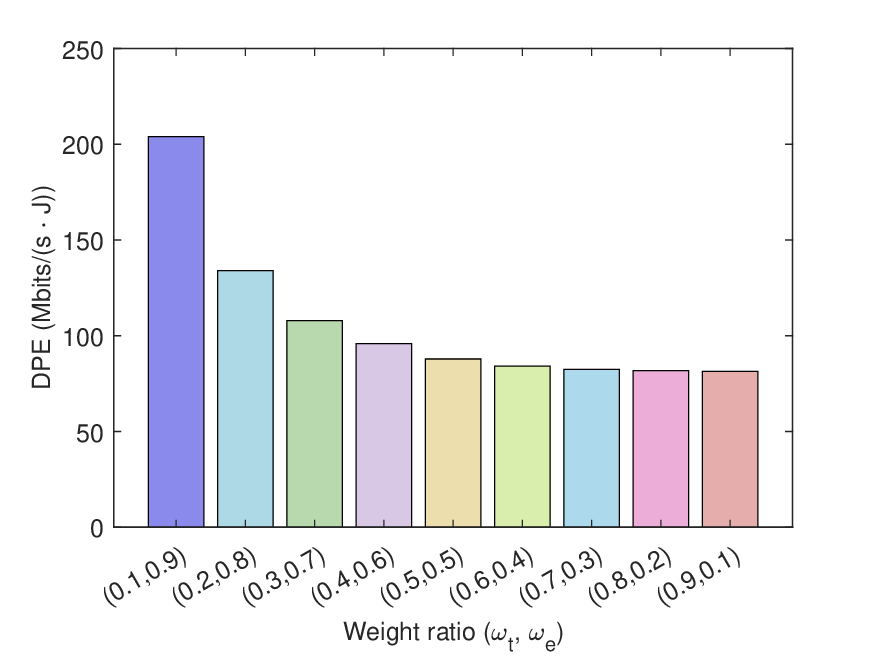}\label{fig.comparison omega}} 
\subfigure[Performance comparisons of different preference weights.]{\includegraphics[width=.33\textwidth]{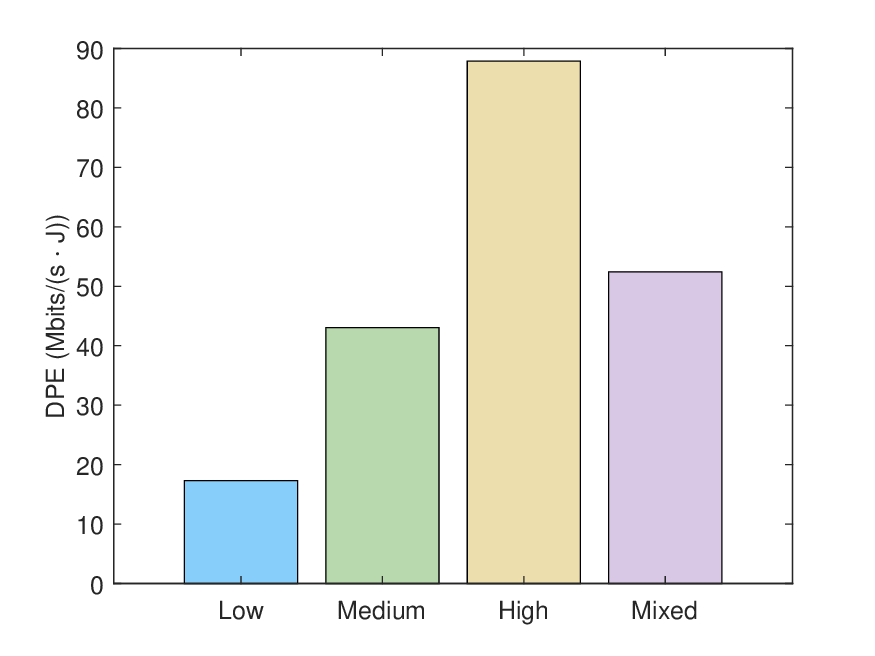}\label{fig.comparison preference}}  
\caption{Performance comparison of different communication and computation resources and heterogeneous settings.}
\end{figure*}
\section{Simulation Results}\label{sec.simulation_results}
In this section, we give the default simulation settings and then present the numerical results and analysis.

\textbf{Default settings.}
We employ a circular network topology with a radius of 1000 meters \cite{yang2015performance}, comprising ten users and two servers randomly positioned. The large-scale fading between user $n$ and server $m$, denoted as $h_{n,m}$, follows the model $128.1+37.6 \log_{10}d_{n,m}$, with $d_{n,m}$ representing the Euclidean distance between them. Rayleigh fading is assumed for small-scale fading. The Gaussian noise power spectral density $\sigma^2$, is set at $-134$dBm, and each server's total bandwidth $b_m$, is 10 MHz. Mobile users have a maximum transmit power $p_{n}$, of 0.2 W, and their maximum computational capacity $f_{n}$, is 1 GHz, while servers' $f_{m}$, have a capacity of 20 GHz. Both mobile users and servers process data at 279.62 CPU cycles per bit ($\eta_n$ and $\eta_m$), and servers require 737.5 CPU cycles per bit for blockchain block generation. The effective switched capacitance for both $\kappa_n$ and $\kappa_m$ is $10^{-27}$. Data sizes for mobile users vary between 500 KB and 2000 KB, determined by scaling uniformly distributed pseudorandom values within this range. The block size $S_b$ is 8 MB, and the data rate for wired server links $R_{m}$ is 15 Mbps. The ratio of data size change between MEC and blockchain tasks $\omega_b$ is fixed at 1. For delay and energy consumption parameters ($\omega_t$ and $\omega_e$), values of 0.5 each are used. DPE preference weights, $c_n$ and $c_{n,m}$ are set at $1 \times \frac{1}{5\times10^{5}}$ to maintain DPE values between 1 and 100. The penalty weight $\varpi$ is set as 175. The Mosek CVX optimization tool is utilized for the simulations.

\begin{table}[t]
\centering
\caption{Convergence of the DAUR Algorithm. This table shows that the number of rounds till convergence is often in single digits.}
\label{tab:DAURConvergence}
\begin{tabular}{@{}lccc@{}}
\toprule
$N$, $M$ & \textbf{10, 2} & \textbf{20, 3} & \textbf{30, 4} \\
\midrule
Outer iter. rounds & 1 & 1 & 1 \\
QCQP rounds & 1 & 1 & 1 \\
FP rounds & 3 (26+25+25 iter.) & 2 (29+28 iter.) & 1 (30 iter.) \\
QCQP iter. & 21 & 25 & 41 \\
QCQP time & 0.19s & 13.88s & 228.78s \\
FP time & 0.03s & 0.10s & 0.08s \\
Total time & 0.22s & 13.98s & 228.86s \\
\bottomrule
\end{tabular}
\end{table}
In Fig. \ref{fig.convergence of FP} and Fig. \ref{fig.convergence of QCQP}, when examining the FP and QCQP methods individually, they demonstrate convergence within just 8 iterations. 
Note that the term ``iterations'' in Fig. \ref{fig.convergence_qcqp_fp} refers to the internal optimization iterations executed by the Mosek solver during each round of QCQP or FP processing. Since CVX does not expose intermediate solution values during its internal optimization process, we use the primal objective value reported by CVX as a surrogate for the DPE value in order to visualize the convergence trend across iterations.

In Table \ref{tab:DAURConvergence}, the DAUR algorithm showcases effective convergence across various user-server configurations in a blockchain-enabled Metaverse wireless communication system. For the (10, 2) user-server setup, it achieves convergence in just one total outer iteration, utilizing one QCQP round and three FP rounds. The QCQP method takes 21 iterations and 0.19 seconds, while the FP method cumulatively takes 76 iterations (computed from 26+25+25) and 0.03 seconds, leading to a total time of 0.22 seconds. In a more complex (20, 3) configuration, the algorithm still maintains its efficiency with one total outer iteration, one QCQP round of 25 iterations (13.88 seconds), and two FP rounds totaling 57 iterations (0.10 seconds), culminating in 13.98 seconds overall. For the largest tested setup of (30, 4), the DAUR algorithm consistently exhibits its robust convergence capability, completing within one outer iteration, one QCQP round of 41 iterations (228.78 seconds), and one FP round of 30 iterations (0.08 seconds), summing up to 228.86 seconds in total.

Compared to the conference version, the journal version achieves an approximate $1.6\%$ improvement in the Problem $\mathbb{P}_1$ objective value with reduced computational complexity. Notably, the number of QCQP iterations required in the journal version is actually reduced across all tested configurations (e.g., from 30–58 iterations in the conference version to 21–41 iterations in the journal version). This reduction reflects the improved efficiency of the refined solver setup, which accelerates convergence while simultaneously achieving better solution quality. Consequently, the total computational time remains comparable or even slightly lower, demonstrating a highly favorable trade-off: better performance, faster convergence, and sustained practical deployability for the Metaverse wireless communication systems.

\textbf{Comparison of convergence behavior with other baselines.} 
In Fig.~\ref{fig.convergence_iteration}, the red box highlights the points corresponding to the QCQP category, while the blue box marks those in the FP category. The AAUCO method reaches convergence within a single QCQP round. GUCRO converges using two FP rounds. In contrast, DAUR undergoes one QCQP round followed by three FP rounds to achieve convergence. In Fig.~\ref{fig.convergence_time}, RUCAA and GUCAA converge instantly within one round, but their achieved DPEs are relatively low (80.78 and 80.38, respectively), with negligible runtimes of less than 0.0012 seconds. AAUCO improves upon this by reaching a DPE of 83.25 in 0.19 seconds but shows no further progress. GUCRO slightly outperforms AAUCO with DPEs up to 84.82 across two rounds, consuming 0.08 seconds in total. DAUR demonstrates the most significant improvement, progressively enhancing DPE from 83.25 to 87.87 over four rounds with a total runtime of 0.22 seconds. These results indicate that DAUR achieves the highest DPE while maintaining a reasonable computational cost.

\textbf{Performance comparison with other baselines.}
We choose the following baselines --- RUCAA: random user connection with average resource allocation; GUCAA: greedy user connection with average resource
allocation; AAUCO: average resource allocation with user connection
optimization; GUCRO: greedy user connection with resource allocation optimization. Note that user connection optimization and resource allocation refer to the QCQP and FP methods, respectively.

The DAUR algorithm outperforms other baseline methods in DPE performance in Fig. \ref{fig.DAUR performance}, achieving 87.87 M bits/(s $\cdot$ J). This surpasses RUCAA (80.78), GUCAA (80.38), AAUCO (83.25), and GUCRO (84.82), demonstrating its superior efficiency in resource allocation and user connection optimization, attributed to its effective integration of QCQP and FP methods.

\textbf{Performance comparison of different communication and computation resources.}
Across the server available bandwidth range of 1 to 10 MHz, the DAUR algorithm consistently outshines other methods in DPE performance in Fig. \ref{fig.comparison bandwidth}, starting at 83.35 and leading a peak of 87.87, notably surpassing other methods in efficiency enhancement. With the increase in communication bandwidth, the DPE of each method shows an increasing trend, which indicates that under the existing configuration, DPE is positively correlated with communication bandwidth.

In the simulation with varying server computation capacities from 2 GHz to 20 GHz in Fig. \ref{fig.comparison server frequency}, the DAUR algorithm exhibits superior DPE performance, steadily increasing from 78.99 to 87.87. This outperformance is consistent across all capacity levels, surpassing the other methods.
In the simulation where user computing capacity varies from 0.1 GHz to 1 GHz in Fig. \ref{fig.comparison user frequency}, the DAUR algorithm consistently demonstrates the highest DPE performance, increasing from 26.29 to 87.87. This trend indicates DAUR's effective handling of increased computing capacities, outperforming other methods at every step.
In summary, the DAUR algorithm outperforms other baselines in varying computing capacities, demonstrating its effective DPE performance and resource management in both server and user capacity simulations.

In Fig. \ref{fig.comparison user transmit power} where the maximum transmit power of mobile users ($p_n$) is adjusted from 0.02 W to 0.2 W, the DAUR algorithm demonstrates superior and consistent DPE performance, starting at 85.52 and peaking at 87.87. 
From Figs. \ref{fig.comparison server frequency}, \ref{fig.comparison user frequency}, and \ref{fig.comparison user transmit power}, we know that DAUR's approach is still better than the baselines using partial optimization, indicating that changes in network communication and computing resources require us to re-optimize user connectivity and resource allocation simultaneously.

\textbf{Performance comparison of heterogeneous settings.}
When adjusting the balance between delay and energy consumption weights ($\omega_t$ and $\omega_e$) in the DAUR algorithm, a notable pattern emerges in DPE values in Fig. \ref{fig.comparison omega}. As the emphasis shifts progressively from delay-centric $(0.1, 0.9)$ to energy-centric $(0.9, 0.1)$, there's an initial steep decrease in DPE from 203.99 to a balanced point of 87.87 at $(0.5, 0.5)$, followed by a more gradual reduction, eventually reaching 81.38. This trend underscores the algorithm's sensitivity to the trade-off between delay and energy efficiency, significantly influencing its performance.

We consider four preference parameter setting cases: 1. low preference: set $c_n$ and $c_{n,m}$ as $0.2/(5\times10^{5})$;
2. medium preference: set $c_n$ and $c_{n,m}$ as $0.5/(5\times10^{5})$;
3. high preference: set $c_n$ and $c_{n,m}$ as $1/(5\times10^{5})$;
4. mixed preference: set $c_n$ and $c_{n,m}$ as $a/(5\times10^{5})$,
where $a$ is a random value uniformly taken from [0, 1]. Fig. \ref{fig.comparison preference} presents the results of the DAUR algorithm under different preference settings for a network with 10 users and 2 servers. In the low preference scenario, the DPE is at its lowest (17.29), indicating minimal efficiency. The medium preference setting improves to a DPE of 43.03, while the high preference achieves the best performance at 87.87, showing the highest efficiency. The mixed preference setting yields a moderate DPE of 52.42, indicating balanced resource distribution. These results show how different preference intensities directly impact the system's data processing efficiency.

\section{Conclusion and Future Direction}\label{sec.conclusion}
In conclusion, this paper presents a potential metric in blockchain-empowered Metaverse wireless communication systems, the concept of data processing efficiency (DPE), and develops a novel DAUR algorithm for optimizing user association and resource allocation. Our approach can transform complex DPE optimization into convex problems, demonstrating superior efficiency over traditional methods through extensive numerical analysis. The potential application of this work is that the Metaverse server helps users conduct non-fungible token (NFT) tasks and considers the maximization of data processing efficiency. In the current work, the algorithm is still a centralized method, which may increase the risk of exposing the users' and servers' private information. As future work, we plan to explore decentralized implementations of the DAUR algorithm to enhance privacy preservation and system robustness. Additionally, we will consider incorporating task utility into the DPE metric, allowing it to better reflect the actual value and importance of the processed tasks. Moreover, we will investigate leveraging polynomial optimization techniques at each reformulation stage to seek global optimality, potentially further improving the solution quality.



\bibliographystyle{IEEEtran}
\bibliography{ref}

\begin{appendices}
\section{Proof of \textbf{Lemma \ref{lemma_p1top2}}}\label{append_lemma_p1top2}
\begin{proof}
Define new auxiliary variables $\vartheta_n^{(u)}$ and $\vartheta_{n,m}^{(s)}$. Let $\vartheta_n^{(u)} \leq \frac{c_n(1-\varphi_n)d_n}{cost^{(u)}_n}$ and $\vartheta_{n,m}^{(s)} \leq \frac{c_{n,m}x_{n,m}\varphi_nd_n}{cost^{(s)}_{n,m}}$. According to Equations (\text{\ref{eq.cost_u}}) and (\text{\ref{eq.cost_s}}), we obtain that
\begin{talign}
    &cost^{(u)}_n = \omega_t T^{(up)}_n + \omega_e E^{(up)}_n, cost^{(u)}_n \leq \frac{c_n(1-\varphi_n)d_n}{\vartheta_n^{(u)}} \nonumber \\
    &\Rightarrow \omega_t T^{(up)}_n + \omega_e E^{(up)}_n \leq \frac{c_n(1-\varphi_n)d_n}{\vartheta_n^{(u)}} \nonumber \\
    &\Rightarrow \omega_t T^{(up)}_n + \omega_e \kappa_n(1-\varphi_n)d_n\eta_n(\psi_nf_n)^2 - \frac{c_n(1-\varphi_n)d_n}{\vartheta_n^{(u)}}\leq 0, \label{eq.cost_u2}
\end{talign}
\begin{talign}
    &cost^{(s)}_{n,m} = \omega_t (T^{(ut)}_{n,m} + T^{(sp)}_{n,m} + T^{(sg)}_{n,m} + T^{(bp)}_{n,m} + T^{(sv)}_{n,m}) \nonumber \\ 
    &+ \omega_e (E^{(ut)}_{n,m} + E^{(sp)}_{n,m} + E^{(sg)}_{n,m}),  cost^{(s)}_{n,m} \leq \frac{c_{n,m}x_{n,m}\varphi_nd_n}{\vartheta_{n,m}^{(s)}} \nonumber \\
    &\Rightarrow \omega_t (T^{(ut)}_{n,m} + T^{(sp)}_{n,m} + T^{(sg)}_{n,m} + T^{(bp)}_{n,m} + T^{(sv)}_{n,m}) \nonumber \\ 
    &+ \omega_e (E^{(ut)}_{n,m} + E^{(sp)}_{n,m} + E^{(sg)}_{n,m})\leq \frac{c_{n,m}x_{n,m}\varphi_nd_n}{\vartheta_{n,m}^{(s)}} \nonumber \\
    &\Rightarrow \omega_t(T^{(ut)}_{n,m} + T^{(sp)}_{n,m} + T^{(sg)}_{n,m} + T^{(bp)}_{n,m} + T^{(sv)}_{n,m}) \nonumber \\ 
    &+ \omega_e \{\frac{x_{n,m}\rho_np_n\varphi_nd_n}{r_{n,m}} + \kappa_mx_{n,m}\varphi_nd_n\eta_m(\gamma_{n,m}\zeta_{n,m}f_m)^2 \nonumber \\ 
    &+\!\! \kappa_mx_{n,m}\varphi_nd_n\eta_m \omega_b [(1\!\!-\!\!\gamma_{n,m})\zeta_{n,m}f_m]^2\} \!\!-\!\! \frac{c_{n,m}x_{n,m}\varphi_nd_n}{\vartheta_{n,m}^{(s)}} \leq 0,\label{eq.cost_s2}
\end{talign}
where ``$\Rightarrow$'' means ``imply''. Next, we leverage two new auxiliary variables $T^{(u)}_n$ and $T^{(s)}_{n,m}$ to replace the delay in Equations (\text{\ref{eq.cost_u2}}) and (\text{\ref{eq.cost_s2}}). With the introduction of these two variables, there would be four more new constraints:
\begin{talign}
    &\omega_t T^{(u)}_n + \omega_e \kappa_n(1-\varphi_n)d_n\eta_n(\psi_nf_n)^2 - \frac{c_n(1-\varphi_n)d_n}{\vartheta_n^{(u)}}\leq 0,\\
    &\omega_t T^{(s)}_{n,m} \!+ \!\omega_e \{\frac{x_{n,m}\rho_np_n\varphi_nd_n}{r_{n,m}} \!+ \!\kappa_mx_{n,m}\varphi_nd_n\eta_m(\gamma_{n,m}\zeta_{n,m}f_m)^2 \nonumber \\ 
    &+ \kappa_mx_{n,m}\varphi_nd_n\eta_m \omega_b [(1-\gamma_{n,m})\zeta_{n,m}f_m]^2\} \nonumber \\
    &- \frac{c_{n,m}x_{n,m}\varphi_nd_n}{\vartheta_{n,m}^{(s)}} \leq 0, \\
    &T^{(up)}_n \leq T^{(u)}_n, \\
    &T^{(ut)}_{n,m} + T^{(sp)}_{n,m} + T^{(sg)}_{n,m} + T^{(bp)}_{n,m} + T^{(sv)}_{n,m} \leq T^{(s)}_{n,m}.
\end{talign}
Let $\bm{T}:=\{\bm{T^{(u)}},\bm{T^{(s)}}\}$ and $\bm{\vartheta}:=\{\bm{\vartheta^{(u)}},\bm{\vartheta^{(s)}}\}$. Then, the Problem $\mathbb{P}_{1}$ can be relaxed into the \mbox{Problem $\mathbb{P}_{2}$}.

Thus, \textbf{Lemma \ref{lemma_p1top2}} is proven.
\end{proof}

\section{Proof of \textbf{Lemma \ref{lemma_p2top3}}}\label{append_lemma_p2top3}
\begin{proof}
Here we analyze part of the KKT condition of Problem $\mathbb{P}_{2}$ to facilitate our subsequent analysis. We introduce non-negative variables $\alpha^{(u)}_n$ and $\alpha^{(s)}_{n,m}$ as the multipliers. Let $\bm{\alpha^{(u)}}:=[\alpha^{(u)}_n]|_{n\in\mathcal{N}}$, $\bm{\alpha^{(s)}}:=[\alpha^{(s)}_{n,m}]|_{n\in\mathcal{N},m\in\mathcal{M}}$, and $\bm{\alpha}:=\{\bm{\alpha^{(u)}}, \bm{\alpha^{(s)}}\}$. The Lagrangian function is given as follows:
\begin{talign}
    &L_{\mathbb{P}_{2}}(\bm{x},\bm{\varphi},\bm{\gamma},\bm{\phi},\bm{\rho},\bm{\zeta},\bm{\psi}, \bm{\vartheta}, \bm{T}, \bm{\alpha}) \nonumber \\
    &= - \sum_{n \in \mathcal{N}} \vartheta_n^{(u)} -\sum_{n \in \mathcal{N}} \sum_{m \in \mathcal{M}} \vartheta_{n,m}^{(s)} \nonumber \\
    &+ \sum_{n \in \mathcal{N}} \alpha^{(u)}_n \cdot [\vartheta_n^{(u)}cost^{(u)}_n - c_n(1 - \varphi_n)d_n] \nonumber \\
    &+ \sum_{n \in \mathcal{N}}\sum_{m \in \mathcal{M}}\alpha^{(s)}_{n,m} \cdot [\vartheta_{n,m}^{(s)}cost^{(s)}_{n,m} - c_{n,m}x_{n,m}\varphi_n d_n] \nonumber \\
    &+ \hat{L}_{\mathbb{P}_{2}},
\end{talign}
where $\hat{L}_{\mathbb{P}_{2}}$ is the remaining Lagrangian terms that we don't care about. Next, we analyze some stationarity and complementary slackness properties of $L_{\mathbb{P}_{2}}$.

\textbf{Stationarity:}
\begin{talign}
    &\frac{\partial L_{\mathbb{P}_{2}}}{\partial \vartheta_n^{(u)}} = -1 + \alpha^{(u)}_n cost^{(u)}_n=0, \forall n \in \mathcal{N},\\
    &\frac{\partial L_{\mathbb{P}_{2}}}{\partial \vartheta_{n,m}^{(s)}} = -1 + \alpha^{(s)}_{n,m} cost^{(s)}_{n,m}=0, \forall n \in \mathcal{N}, \forall m \in \mathcal{M}.
\end{talign}

\textbf{Complementary slackness:}
\begin{talign}
    &\alpha^{(u)}_n \cdot [\vartheta_n^{(u)}cost^{(u)}_n - c_n(1 - \varphi_n)d_n] = 0,\forall n \in \mathcal{N},\\
    &\alpha^{(s)}_{n,m} \cdot (\vartheta_{n,m}^{(s)}cost^{(s)}_{n,m} - c_{n,m}x_{n,m}\varphi_n d_n) = 0, \nonumber \\
    &\forall n \in \mathcal{N}, \forall m \in \mathcal{M}.
\end{talign}
Thus, at KKT points of Problem $\mathbb{P}_{2}$, we can obtain that 
\begin{talign}
    &\alpha^{(u)}_n = \frac{1}{cost^{(u)}_n},\\
    &\alpha^{(s)}_{n,m} = \frac{1}{cost^{(s)}_{n,m}},\\
    &\vartheta_n^{(u)} = \frac{c_n(1 - \varphi_n)d_n}{cost^{(u)}_n},\\
    &\vartheta_{n,m}^{(s)} = \frac{c_{n,m}x_{n,m}\varphi_n d_n}{cost^{(s)}_{n,m}}.
\end{talign}
Based on the above discussion, Problem $\mathbb{P}_{2}$ can be transformed into Problem $\mathbb{P}_{3}$.

\textbf{Lemma \ref{lemma_p2top3}} holds.
\end{proof}

\section{Proof of \textbf{Lemma \ref{lemma_p4top5}}}\label{append_lemma_p4top5}
\begin{proof}
We introduce one auxiliary variable $\upsilon^{(s)}_{n,m}$, where
\begin{talign}
    \upsilon^{(s)}_{n,m} = \frac{1}{2x_{n,m}\rho_np_n\varphi_nd_nr_{n,m}}.
\end{talign}
Then, $cost^{(s)}_{n,m}$ can be rewritten as:
\begin{talign}
    &\widetilde{cost}^{(s)}_{n,m} = \omega_t T^{(s)}_{n,m} + \omega_e \{(x_{n,m}\rho_np_n\varphi_nd_n)^2\upsilon^{(s)}_{n,m} + \frac{1}{4r_{n,m}^2\upsilon^{(s)}_{n,m}}\nonumber \\ 
    &+ \kappa_mx_{n,m}\varphi_nd_n\eta_m(\gamma_{n,m}\zeta_{n,m}f_m)^2 \nonumber \\ 
    &+ \kappa_mx_{n,m}\varphi_nd_n\eta_m \omega_b [(1-\gamma_{n,m})\zeta_{n,m}f_m]^2\}. \label{eq.vartheta_s_new}
\end{talign}
Since $x_{n,m}, \varphi_n, \gamma_{n,m}, \upsilon^{(s)}_{n,m}$ is given, Equation (\text{\ref{eq.vartheta_s_new}}) is a convex function.

Let $\chi(\rho_n) = x_{n,m}\rho_np_n\varphi_nd_n$ and $\varsigma(\phi_{n,m},\rho_n) = r_{n,m}$, where $r_{n,m} = \phi_{n,m}b_m\log_2(1 + \frac{\rho_np_ng_{n,m}}{\sigma^2\phi_{n,m}b_m})$. It's easy to know that $\chi(\rho_n)$ is convex of $\rho_n$ and $\varsigma(\phi_{n,m},\rho_n)$ is jointly concave of $(\phi_{n,m},\rho_n)$. We define two following functions:
\begin{talign}
    &\mathcal{F}(\rho_n,\phi_{n,m},\zeta_{n,m},T^{(s)}_{n,m}) = \omega_t T^{(s)}_{n,m} + \omega_e \{\frac{\chi(\rho_n)}{\varsigma(\phi_{n,m},\rho_n)} \nonumber \\ 
    &+ \kappa_mx_{n,m}\varphi_nd_n\eta_m(\gamma_{n,m}\zeta_{n,m}f_m)^2 \nonumber \\ 
    &+ \kappa_mx_{n,m}\varphi_nd_n\eta_m \omega_b [(1-\gamma_{n,m})\zeta_{n,m}f_m]^2\},
\end{talign}
\begin{talign}
    &\mathcal{G}(\rho_n,\phi_{n,m},\zeta_{n,m},T^{(s)}_{n,m}) = \omega_t T^{(s)}_{n,m} \nonumber \\ 
    &+ \omega_e \{\chi(\rho_n)^2\upsilon^{(s)}_{n,m} + \frac{1}{4\varsigma(\phi_{n,m},\rho_n)^2\upsilon^{(s)}_{n,m}}\nonumber \\ 
    &+ \kappa_mx_{n,m}\varphi_nd_n\eta_m(\gamma_{n,m}\zeta_{n,m}f_m)^2 \nonumber \\ 
    &+ \kappa_mx_{n,m}\varphi_nd_n\eta_m \omega_b [(1-\gamma_{n,m})\zeta_{n,m}f_m]^2\}.
\end{talign}
The partial derivative of $T^{(s)}_{n,m}$ is given by
\begin{talign}
    \frac{\partial\mathcal{F}(\rho_n,\phi_{n,m},\zeta_{n,m},\vartheta^{(s)}_{n,m},T^{(s)}_{n,m})}{\partial T^{(s)}_{n,m}} = \omega_t,\\
    \frac{\partial\mathcal{G}(\rho_n,\phi_{n,m},\zeta_{n,m},\vartheta^{(s)}_{n,m},T^{(s)}_{n,m})}{\partial T^{(s)}_{n,m}} = \omega_t.
\end{talign}
We can easily get that
\begin{talign}
    \frac{\partial\mathcal{F}(\rho_n,\phi_{n,m},\zeta_{n,m},\vartheta^{(s)}_{n,m},T^{(s)}_{n,m})}{\partial T^{(s)}_{n,m}} = \frac{\partial\mathcal{G}(\rho_n,\phi_{n,m},\zeta_{n,m},\vartheta^{(s)}_{n,m},T^{(s)}_{n,m})}{\partial T^{(s)}_{n,m}}.
\end{talign}
The partial derivative of $\zeta_{n,m}$ is
\begin{talign}
    &\frac{\partial \mathcal{F}(\rho_n,\phi_{n,m},\zeta_{n,m},\vartheta^{(s)}_{n,m},T^{(s)}_{n,m})}{\partial \zeta_{n,m}} = 2\kappa_mx_{n,m}\varphi_nd_n\eta_m\zeta_{n,m}(\gamma_{n,m}f_m)^2 \nonumber \\
    &+ 2\kappa_mx_{n,m}\varphi_nd_n\eta_m \omega_b \zeta_{n,m}[(1-\gamma_{n,m})f_m]^2,\\
    &\frac{\partial \mathcal{G}(\rho_n,\phi_{n,m},\zeta_{n,m},\vartheta^{(s)}_{n,m},T^{(s)}_{n,m})}{\partial \zeta_{n,m}} = 2\kappa_mx_{n,m}\varphi_nd_n\eta_m\zeta_{n,m}(\gamma_{n,m}f_m)^2 \nonumber \\
    &+ 2\kappa_mx_{n,m}\varphi_nd_n\eta_m \omega_b \zeta_{n,m}[(1-\gamma_{n,m})f_m]^2.
\end{talign}
We find that
\begin{talign}
    \frac{\partial \mathcal{F}(\rho_n,\phi_{n,m},\zeta_{n,m},\vartheta^{(s)}_{n,m},T^{(s)}_{n,m})}{\partial \zeta_{n,m}} = \frac{\partial \mathcal{G}(\rho_n,\phi_{n,m},\zeta_{n,m},\vartheta^{(s)}_{n,m},T^{(s)}_{n,m})}{\partial \zeta_{n,m}}.
\end{talign}
The partial derivative of $\rho_n$ is shown as follows:
\begin{talign}
    &\frac{\partial \mathcal{F}(\rho_n,\phi_{n,m},\zeta_{n,m},\vartheta^{(s)}_{n,m},T^{(s)}_{n,m})}{\partial \rho_n} \nonumber \\ 
    &= \omega_e\frac{\frac{\partial \chi(\rho_n)}{\partial \rho_n}\varsigma(\phi_{n,m},\rho_n) - \chi(\rho_n) \frac{\partial \varsigma(\phi_{n,m},\rho_n)}{\partial \rho_n}}{\varsigma(\phi_{n,m},\rho_n)^2},\\
    &\frac{\partial \mathcal{G}(\rho_n,\phi_{n,m},\zeta_{n,m},\vartheta^{(s)}_{n,m},T^{(s)}_{n,m})}{\partial \rho_n} \nonumber \\
    &= \omega_e (2\upsilon^{(s)}_{n,m} \chi(\rho_n) \frac{\partial \chi(\rho_n)}{\partial \rho_n} - \frac{1}{2\upsilon^{(s)}_{n,m} \varsigma(\phi_{n,m},\rho_n)^3}\frac{\partial \varsigma(\phi_{n,m},\rho_n)}{\partial \rho_n}).
\end{talign}
When $\upsilon^{(s)}_{n,m} = \frac{1}{2x_{n,m}\rho_np_n\varphi_nd_nr_{n,m}}$, we can get that
\begin{talign}
    \frac{\partial \mathcal{F}(\rho_n,\phi_{n,m},\zeta_{n,m},\vartheta^{(s)}_{n,m},T^{(s)}_{n,m})}{\partial \rho_n} = \frac{\partial \mathcal{G}(\rho_n,\phi_{n,m},\zeta_{n,m},\vartheta^{(s)}_{n,m},T^{(s)}_{n,m})}{\partial \rho_n}.
\end{talign}
The partial derivative of $\phi_{n,m}$ is
\begin{talign}
    &\frac{\partial \mathcal{F}(\rho_n,\phi_{n,m},\zeta_{n,m},\vartheta^{(s)}_{n,m},T^{(s)}_{n,m})}{\partial \phi_{n,m}} = -\frac{\omega_e\chi(\rho_n)}{\varsigma(\phi_{n,m},\rho_n)^2}\frac{\partial \varsigma(\phi_{n,m},\rho_n)}{\partial \phi_{n,m}},\\
    &\frac{\partial \mathcal{G}(\rho_n,\phi_{n,m},\zeta_{n,m},\vartheta^{(s)}_{n,m},T^{(s)}_{n,m})}{\partial \phi_{n,m}} = -\frac{\omega_e}{2\upsilon^{(s)}_{n,m}\varsigma(\phi_{n,m},\rho_n)^3}\frac{\partial \varsigma(\phi_{n,m},\rho_n)}{\partial \phi_{n,m}}.
\end{talign}
When $\upsilon^{(s)}_{n,m} = \frac{1}{2x_{n,m}\rho_np_n\varphi_nd_nr_{n,m}}$, we can obtain that
\begin{talign}
    \frac{\partial \mathcal{F}(\rho_n,\phi_{n,m},\zeta_{n,m},\vartheta^{(s)}_{n,m},T^{(s)}_{n,m})}{\partial \phi_{n,m}} =\frac{\partial \mathcal{G}(\rho_n,\phi_{n,m},\zeta_{n,m},\vartheta^{(s)}_{n,m},T^{(s)}_{n,m})}{\partial \phi_{n,m}}.
\end{talign}
Based on the above discussion, we can obtain that
\begin{talign}
     \frac{\partial \mathcal{F}(\rho_n,\phi_{n,m},\zeta_{n,m},\vartheta^{(s)}_{n,m},T^{(s)}_{n,m})}{\partial (\rho_n,\phi_{n,m},\zeta_{n,m},\vartheta^{(s)}_{n,m},T^{(s)}_{n,m})} =\frac{\partial \mathcal{G}(\rho_n,\phi_{n,m},\zeta_{n,m},\vartheta^{(s)}_{n,m},T^{(s)}_{n,m})}{\partial (\rho_n,\phi_{n,m},\zeta_{n,m},\vartheta^{(s)}_{n,m},T^{(s)}_{n,m})}.
\end{talign}
Besides, it's easy to know that 
\begin{talign}
    &\mathcal{F}(\rho_n,\phi_{n,m},\zeta_{n,m},\vartheta^{(s)}_{n,m},T^{(s)}_{n,m})\nonumber \\ 
    &= \mathcal{G}(\rho_n,\phi_{n,m},\zeta_{n,m},\vartheta^{(s)}_{n,m},T^{(s)}_{n,m}),
\end{talign}
when $\upsilon^{(s)}_{n,m} = \frac{1}{2x_{n,m}\rho_np_n\varphi_nd_nr_{n,m}}$. Therefore, the function of $\mathcal{F}(\cdot)$ is the same as that of $\mathcal{G}(\cdot)$. Let $\bm{\upsilon^{(s)}} := [\upsilon^{(s)}_{n,m}|_{\forall n \in \mathcal{N},\forall m \in \mathcal{M}}]$. Problem $\mathbb{P}_{4}$ can be transformed into Problem $\mathbb{P}_{5}$. In Problem $\mathbb{P}_{5}$, if given $\bm{\upsilon^{(s)}}$, it would be a convex optimization problem.

At the $i$-th iteration, if we first fix $\bm{\upsilon}^{\bm{(s)}(i-1)}$, Problem $\mathbb{P}_{5}$ would be a concave optimization problem. Then we optimize $\bm{\phi}^{(i)},\bm{\rho}^{(i)},\bm{\zeta}^{(i)}$, $\bm{\psi}^{(i)}, \bm{T}^{(i)}$. After we obtain the results of them, we then update $\bm{\upsilon}^{\bm{(s)}(i)}$ according to those results. Because the alternative optimization of the Problem $\mathbb{P}_{5}$ is non-decreasing, as $i\rightarrow\infty$, we can finally obtain the optimal solutions of Problem $\mathbb{P}_{5}$ (i.e., $\bm{\phi}^{(\star)}$, $\bm{\rho}^{(\star)}$, $\bm{\zeta}^{(\star)}$, $\bm{\psi}^{(\star)}$, $\bm{T}^{(\star)}$, $\bm{\upsilon}^{\bm{(s)}(\star)}$). We know that $\upsilon^{(s)}_{n,m} = \frac{1}{2x_{n,m}\rho_np_n\varphi_nd_nr_{n,m}}$. Thus, with $\bm{\upsilon}^{\bm{(s)}(\star)}$, we can find $\bm{\phi}^{(\star)},\bm{\rho}^{(\star)},\bm{\zeta}^{(\star)},\bm{\psi}^{(\star)}, \bm{T}^{(\star)}$, which is a stationary point of Problem $\mathbb{P}_{5}$.

\textbf{Lemma \ref{lemma_p4top5}} is proven.
\end{proof}

\section{Proof of \textbf{Lemma \ref{lemma_gamma}}}\label{append_lemma_gamma}
\begin{proof}
    We first analyze $\vartheta_{n,m}^{(s)}cost_{n,m}^{(s)}$, and write the explicit expression of it:
\begin{talign}
    &\vartheta_{n,m}^{(s)}cost_{n,m}^{(s)} \nonumber \\
    &= \vartheta_{n,m}^{(s)}\omega_t T^{(s)}_{n,m} + \vartheta_{n,m}^{(s)}\omega_e \frac{\rho_n p_n d_n}{r_{n,m}}x_{n,m}\varphi_n \nonumber \\
    &+ \vartheta_{n,m}^{(s)}\omega_e \{\kappa_m d_n \eta_m\zeta_{n,m}^2 f_m^2 [\gamma_{n,m}^2 + \omega_b (1-\gamma_{n,m})^2]x_{n,m}\varphi_n\},
\end{talign}
where $\gamma_{n,m}^2 + \omega_b (1-\gamma_{n,m})^2$ is independent of the others except $T^{(s)}_{n,m}$. when $\gamma_{n,m} = \frac{\omega_b}{1+\omega_b}$, $\gamma_{n,m}^2 + \omega_b (1-\gamma_{n,m})^2$ takes the minimum value. In $T^{(s)}_{n,m}$, the terms $T^{(sp)}_{n,m}$, $T^{(sg)}_{n,m}$, and $T^{(sv)}_{n,m}$ are related to $\gamma_{n,m}$. Since $T^{(sv)}_{n,m}$ is generally much smaller than $T^{(sp)}_{n,m}$ and $T^{(sg)}_{n,m}$, we only focus on $T^{(sp)}_{n,m}$ and $T^{(sg)}_{n,m}$ here. It's easy to know that when $\gamma_{n,m} = \frac{1}{1+\omega_b}$, $T^{(sp)}_{n,m} + T^{(sg)}_{n,m}$ takes the minimum value according to basic inequality. Following are the detailed steps:
\begin{talign}
    T^{(sp)}_{n,m} + T^{(sg)}_{n,m} \geq 2\sqrt{T^{(sp)}_{n,m} T^{(sg)}_{n,m}},
\end{talign}
where if and only if $T^{(sp)}_{n,m} =T^{(sg)}_{n,m}$, ``$=$'' can be obtained.
\begin{talign}
    &\quad \quad T^{(sp)}_{n,m} =T^{(sg)}_{n,m},\nonumber \\
    &\Rightarrow \frac{x_{n,m}\varphi_nd_n\eta_m}{\gamma_{n,m}\zeta_{n,m}f_m} = \frac{x_{n,m}\varphi_nd_n\omega_b\eta_m}{(1-\gamma_{n,m})\zeta_{n,m}f_m}, \nonumber \\
    &\Rightarrow \gamma_{n,m} = \frac{1}{1+\omega_b}.
\end{talign}
We set $\omega_b = 1$ and then $\frac{\omega_b}{1+\omega_b} = \frac{1}{1+\omega_b}$, in which case, $cost_{n,m}^{(s)}$ would take the minimum leading to Problem $\mathbb{P}_{7}$ take the maximum value. Based on the above discussion, we can obtain the optimal value of $\gamma_{n,m}$ that $\gamma_{n,m}^\star = \frac{\omega_b}{1+\omega_b}$ or $\frac{1}{1+\omega_b} = \frac{1}{2}$. 

Thus, \textbf{Lemma \ref{lemma_gamma}} is proven.
\end{proof}

\section{Proof of \textbf{Lemma \ref{lemma_p8top9}}}\label{append_lemma_p8top9}
\begin{proof}
Let $\bm{P}_0:=\bm{I}_{NM+N\times N} \bm{I}_{N\rightarrow NM} \text{diag}(\bm{B})\bm{e}_{N+1,NM+N}$ and we can represent $\sum_{n \in \mathcal{N}}\sum_{m \in \mathcal{M}}B_{n,m} x_{n,m} \varphi_n$ as $\bm{Q}^\intercal \bm{P}_0 \bm{Q}$. Let $\bm{W}_0^\intercal:=\bm{A}^\intercal\bm{e}_{1,N}$ and the term $\sum_{n \in \mathcal{N}}A_n \varphi_n$ can be rewritten as $\bm{W}_0^\intercal \bm{Q}$.  Let $P_{0,n}^{(T_u)}:= - \frac{\alpha^{(u)}_n \vartheta_n^{(u)} \omega_t d_n\eta_n}{\psi_nf_n}$, $P_1^{(T_u)}:=\sum_{n \in \mathcal{N}}\frac{\alpha^{(u)}_n \vartheta_n^{(u)} \omega_t d_n\eta_n}{\psi_nf_n}$, $\bm{P}_0^{(T_u)}:=[P_{0,n}^{(T_u)}]|_{n \in \mathcal{N}}$, and ${\bm{P}_2^{(T_u)}}^\intercal:={\bm{P}_0^{(T_u)}}^\intercal\bm{e}_{1,N}$. Then, the constraint (\ref{Tu_constr1}) can be represented by 
\begin{talign}
    {\bm{P}_2^{(T_u)}}^\intercal \bm{Q} + P_1^{(T_u)} \leq T^{(u)}.
\end{talign}
Let 
\begin{talign}
&P_{0,n,m}^{(T_s)}:=\frac{d_n}{r_{n,m}} + \frac{d_n\eta_m}{\gamma_{n,m}\zeta_{n,m}f_m} + \frac{d_n\omega_b\eta_m}{(1-\gamma_{n,m})\zeta_{n,m}f_m},\\
&\bm{P}_0^{(T_s)}:=[P_{0,n,m}^{(T_s)}]|_{n \in \mathcal{N},m \in \mathcal{M}},\\
&P_1^{(T_s)}:=\sum_{n \in \mathcal{N}}\sum_{m \in \mathcal{M}}\frac{S_b}{R_m} + \text{max}_{m^\prime\in\mathcal{M}\setminus\{m\}}\frac{\eta_v}{(1-\gamma_{n,m^\prime})\zeta_{n,m^\prime}f_m^\prime}.
\end{talign}
Similar to $\bm{P}_0$, let 
\begin{talign}
\bm{P}_2^{(T_s)}:= \bm{I}_{NM+N\times N} \bm{I}_{N\rightarrow NM} \text{diag}(\bm{P}_0^{(T_s)})\bm{e}_{N+1,NM+N}.
\end{talign}
Then, the constraint (\ref{Ts_constr1}) can be transformed into
\begin{talign}
    \bm{Q}^\intercal \bm{P}_0^{(T_s)} \bm{Q} + P_1^{(T_s)} \leq T^{(s)}.
\end{talign}
Let $\bm{\phi}:=(\phi_{1,1},\cdots,\phi_{N,M})^\intercal$ and $\bm{\zeta}:=(\zeta_{1,1},\cdots,\zeta_{N,M})^\intercal$. The constraints \text{(\ref{x_constr1_qcqp})}-\text{(\ref{x_zeta_constr_qcqp})} are easy to obtain, which we won't go into details here.

Therefore, \textbf{Lemma \ref{lemma_p8top9}} holds.
\end{proof}

\section{Proof of \textbf{Lemma \ref{lemma_p9top10}}}\label{append_lemma_p9top10}
\begin{proof}
Here we give the expression of matrices $\bm{P}_1$, $\bm{P}_2$, $\bm{P}_3$, $\bm{P}_4$, $\bm{P}_5$, $\bm{P}_6$, $\bm{P}_7$, and $\bm{P}_8$.
\begin{equation}
\bm{P}_1=
\left(
    \begin{array}{cc}
       \bm{P}_0  & \frac{1}{2}\bm{W}_0 \\
        \frac{1}{2}\bm{W}_0^\intercal &  T^{(u)} + T^{(s)} + C
    \end{array}
\right),
\end{equation}
\begin{equation}
\bm{P}_2=
\left(
    \begin{array}{cc}
      \boldsymbol{e}_{i}^\intercal\boldsymbol{e}_{i}   & -\frac{1}{2}\boldsymbol{e}_{i} \\
       -\frac{1}{2}\boldsymbol{e}_{i}^\intercal  & 0
    \end{array}
\right), \forall i \in \{1,\cdots, NM\}
\end{equation}
\begin{align}
\bm{P}_3=
\left(
    \begin{array}{cc}
    \bm{0}_{NM+N \times NM+N}     & \frac{1}{2}(\boldsymbol{e}_{\overline{1},\overline{M}}\boldsymbol{e}_{N+1,NM+N}^\intercal) \\
     \frac{1}{2}(\boldsymbol{e}_{\overline{1},\overline{M}}\boldsymbol{e}_{N+1,NM+N}^\intercal)^\intercal    & -1
    \end{array}
\right)\nonumber, \\ \forall i \in \{1,\cdots, N\}
\end{align}
\begin{equation}
\bm{P}_4=
\left(
    \begin{array}{cc}
      \bm{0}_{NM+N \times NM+N}   & \frac{1}{2}\boldsymbol{e}_{i} \\
       \frac{1}{2}\boldsymbol{e}_{i}^\intercal  & -1
    \end{array}
\right), \forall i \in \{1,\cdots, N\}
\end{equation}
\begin{equation}
\bm{P}_5=
\left(
    \begin{array}{cc}
    \bm{0}_{NM+N \times NM+N}    & \frac{1}{2}\boldsymbol{\phi}\boldsymbol{e}_{N+1,NM+N} \\
    \frac{1}{2}(\boldsymbol{\phi}\boldsymbol{e}_{N+1,NM+N})^\intercal     & -1
    \end{array}
\right),
\end{equation}
\begin{equation}
\bm{P}_6=
\left(
    \begin{array}{cc}
    \bm{0}_{NM+N \times NM+N}     & \frac{1}{2}\boldsymbol{\zeta}\boldsymbol{e}_{N+1,NM+N} \\
    \frac{1}{2}(\boldsymbol{\zeta}\boldsymbol{e}_{N+1,NM+N})^\intercal     & -1
    \end{array}
\right), 
\end{equation}
\begin{equation}
\bm{P}_7=
\left(
    \begin{array}{cc}
    \bm{0}_{NM+N \times NM+N}     & \frac{1}{2}\bm{P}_2^{(T_u)} \\
    \frac{1}{2}{\bm{P}_2^{(T_u)}}^\intercal     & P_1^{(T_u)}
    \end{array}
\right), 
\end{equation}
\begin{equation}
\bm{P}_8=
\left(
    \begin{array}{cc}
    \bm{P}_0^{(T_s)}     & \bm{0}_{NM+N \times 1}\\
    \bm{0}_{1 \times NM+N}     & P_1^{(T_s)}
    \end{array}
\right). 
\end{equation}

\textbf{Lemma \ref{lemma_p9top10}} is proven.
\end{proof}

\end{appendices}

\begin{IEEEbiography}
[{\includegraphics[width=1in,height=1.25in,clip,keepaspectratio]{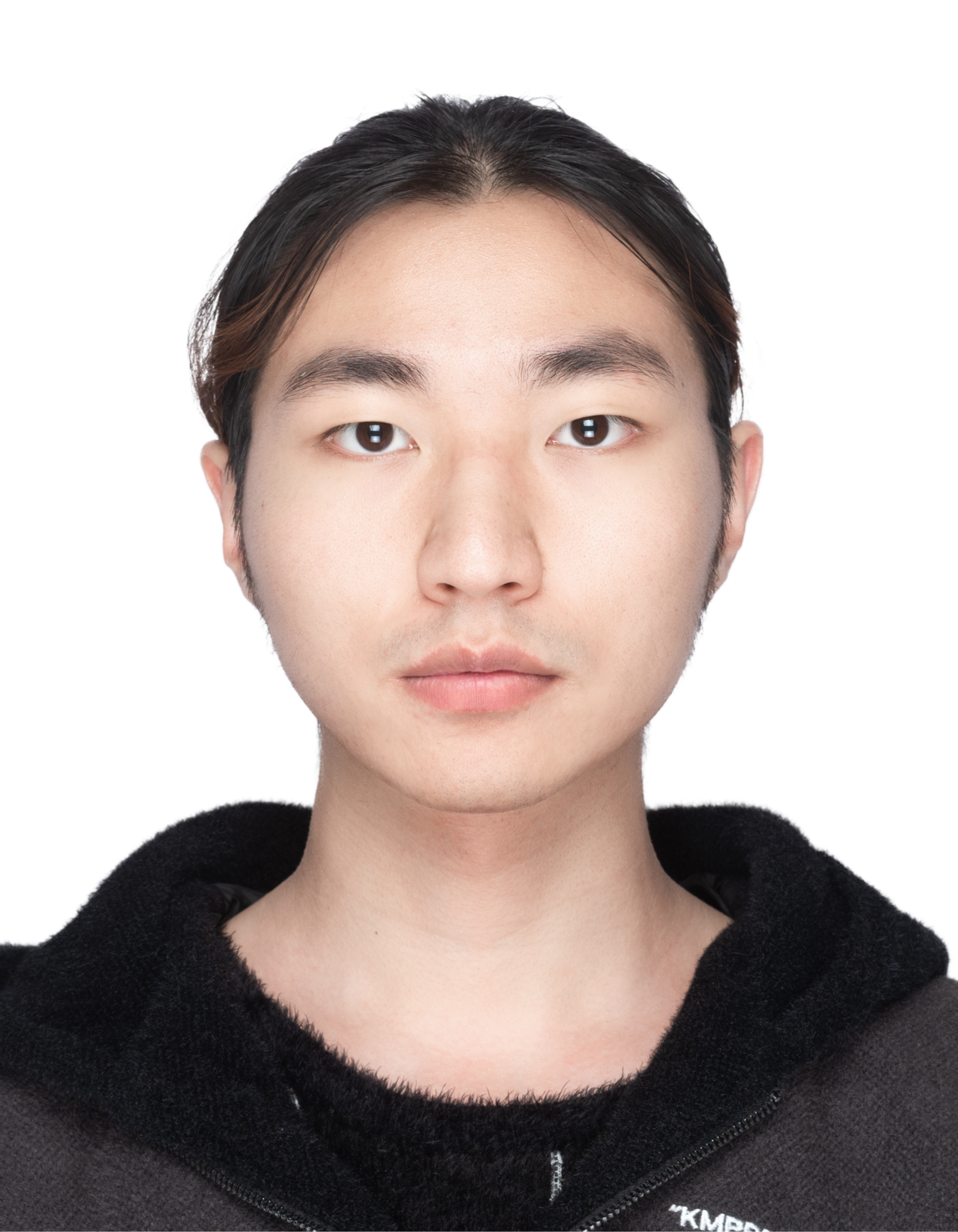}}]
{Liangxin Qian}
(Graduate Student Member, IEEE) received bachelor's and master's degrees in communication engineering from the University of Electronic Science and Technology of China, Chengdu, China, in 2019 and 2022, respectively. He is currently working toward his Ph.D. at the College of Computing and Data Science (CCDS), Nanyang Technological University, Singapore. His research interests include Metaverse, mobile edge computing, and secure communications.
\end{IEEEbiography}

\begin{IEEEbiography}
[{\includegraphics[width=1in,height=1.25in,clip,keepaspectratio]{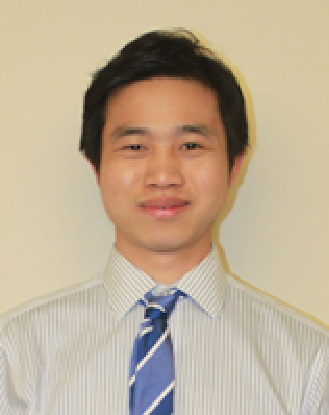}}]
{Jun Zhao} 
(S'10-M'15) is currently an Assistant Professor in the College of Computing and Data Science (CCDS) - formerly known as the School of Computer Science and Engineering (SCSE) at Nanyang Technological University (NTU) in Singapore. He received a PhD degree in May 2015 in Electrical and Computer Engineering from Carnegie Mellon University (CMU) in the USA, affiliating with CMU's renowned CyLab Security \& Privacy Institute, and a bachelor's degree in July 2010 from Shanghai Jiao Tong University in China. Before joining NTU first as a postdoc and then as a faculty member, he was a postdoc at Arizona State University as an Arizona Computing PostDoc Best Practices Fellow. He is currently an Editor of many journals: IEEE Transactions on Information Forensics and Security (TIFS), ACM Distributed Ledger Technologies (DLT), IEEE Internet of Things Journal (IoTJ), and Elsevier Future Generation Computer Systems (FGCS). He has received best editor awards from IEEE journals and best paper awards from conferences/journals. He was also selected as a Best Editor for IEEE Wireless Communications Letters (WCL) and a Best Editor for IEEE Transactions on Cognitive Communications and Networking (TCCN) in 2023.
\end{IEEEbiography}

\end{document}